\documentclass[journal,twocolumn]{IEEEtran}

\IEEEoverridecommandlockouts

\usepackage[cmex10]{amsmath} 
\usepackage{amsfonts,acronym,amssymb,amsthm,dsfont}
\usepackage{graphicx}
\usepackage{url}
\usepackage{cite}
\usepackage{hyperref}
\usepackage{braket}
\usepackage{physics}
\usepackage{float}
\usepackage{dblfloatfix}
\usepackage{enumerate}
\usepackage{lipsum}
\usepackage{diagbox}
\usepackage[dvipsnames]{xcolor}
\usepackage{mathrsfs}
\usepackage{bm}
\usepackage{stackengine}

\makeatletter
\def\blfootnote{\xdef\@thefnmark{}\@footnotetext}
\makeatother

\newcounter{mytempeqcounter}

\allowdisplaybreaks
\allowbreak

\makeatletter
\NewDocumentCommand{\seteqnum}{o}{%
  \IfValueTF{#1}
    {\textup{\tagform@{#1}}}
    {\incr@eqnum \print@eqnum}
}
\NewCommandCopy{\ltxlabel}{\ltx@label}
\makeatother

\newcommand{\calA}{\mathcal{A}}

\newcommand{\calB}{\mathcal{B}}

\newcommand{\calC}{\mathcal{C}}

\newcommand{\calD}{\mathcal{D}}
\newcommand{\bbD}{\mathbb{D}}

\newcommand{\calE}{\mathcal{E}}
\newcommand{\bbE}{\mathbb{E}}

\newcommand{\calH}{\mathcal{H}}

\newcommand{\calI}{\mathcal{I}}
\newcommand{\bbI}{\mathbb{I}}

\newcommand{\calJ}{\mathcal{J}}

\newcommand{\calK}{\mathcal{K}}

\newcommand{\calL}{\mathcal{L}}

\newcommand{\calM}{\mathcal{M}}

\newcommand{\calN}{\mathcal{N}}
\newcommand{\bbN}{\mathbb{N}}

\newcommand{\calP}{\mathcal{P}}
\newcommand{\bbP}{\mathbb{P}}

\newcommand{\calR}{\mathcal{R}}
\newcommand{\bbR}{\mathbb{R}}

\newcommand{\calS}{\mathcal{S}}

\newcommand{\calU}{\mathcal{U}}

\newcommand{\calX}{\mathcal{X}}

\newcommand{\calZ}{\mathcal{Z}}

\newcommand\ubar[1]{\stackunder[1.1pt]{$#1$}{\rule{1.2ex}{.08ex}}}
\newcommand{\pr}[1]{\ensuremath{\left({#1}\right)}}%

\DeclareMathOperator{\tra}{Tr}
\DeclareMathOperator{\D}{D}

\DeclareMathOperator{\C}{C}
\DeclareMathOperator{\F}{F}
\DeclareMathOperator{\Pu}{P}

\newcommand{\den}[2]{\ensuremath{\ket{#1}{\hspace{-1.6mm}}\bra{#2}}}

\newtheorem{theorem}{Theorem}
\newtheorem{corollary}{Corollary}
\newtheorem{definition}{Definition}

\newtheorem{remark}{Remark}
\newtheorem{spc}{Special Case}

\newtheorem{lemma}{Lemma}[]

\newcommand{\indic}[1]{\ensuremath{\mathds{1}}}
\newcommand{\card}[1]{\ensuremath{\left|{#1}\right|}}   

\newcommand{\sbra}[2]{\ensuremath{\left[{#1}{\,:\,}{#2}\right]}}%
\newcommand{\sbr}[1]{\ensuremath{\left[{#1}\right]}}%

\acrodef{ACDIS}[ACDIS]{Adaptive Communication Decision and Information Systems}
\acrodef{AEP}{Asymptotic Equipartition Property}
\acrodef{AoA}{Angle of Arrival}
\acrodef{AWGN}{Additive White Gaussian Noise}
\acrodef{AVC}[AVC]{Arbitrarily Varying Channel}
\acrodef{BER}{Bit-Error-Rate}
\acrodef{BEC}{Binary Erasure Channel}
\acrodef{BPSK}{Binary Phase-Shift Keying}
\acrodef{BSC}{Binary Symmetric Channel}
\acrodef{RV}{Random Variable}
\acrodefplural{RV}{Random Variables}
\acrodef{JTE}{Joint Typical Encoder}
\acrodef{BICM}[BICM]{Bit-Interleaved Coded-Modulation}
\acrodef{CDF}[CDF]{Cumulative Distribution Function}
\acrodef{CGF}[CGF]{Cumulant Generating Function}
\acrodef{CLT}[CLT]{Central Limit Theorem}
\acrodef{CSI}[CSI]{Channel State Information}
\acrodef{DMC}[DMC]{Discrete Memoryless Channel}
\acrodef{DMS}[DMS]{Discrete Memoryless Source}
\acrodef{ERM}[ERM]{Empirical Risk Minimization}
\acrodef{FER}[FER]{Frame Error Rate}
\acrodef{ICA}[ICA]{Independent Component Analysis}
\acrodef{iid}[i.i.d.]{independent and identically distributed}
\acrodef{IoT}[IoT]{Internet of Things}
\acrodef{KKT}[KKT]{Karush-Kuhn Tucker}
\acrodef{LASSO}[LASSO]{Least Absolute Shrinkage and Selection Operator}
\acrodef{LPD}[LPD]{Low Probability of Detection}
\acrodef{LDPC}[LDPC]{Low-Density Parity-Check}
\acrodef{LLMS}[LLMS]{Linear Least Mean Square}
\acrodef{LMS}[LMS]{Least Mean Square}
\acrodef{MAC}[MAC]{multiple-access channel}
\acrodef{MGF}[MGF]{Moment Generating Function}
\acrodef{MLC}[MLC]{Multi-Level Coding}
\acrodef{MLE}[MLE]{Maximum Likelihood Estimate}
\acrodef{MIMO}[MIMO]{Multiple-Input Multiple-Output}
\acrodef{MISO}{Multiple-Input Single-Output}
\acrodef{MSD}[MSD]{Multi-Stage Decoding}
\acrodef{MMSE}[MMSE]{Minimum Mean-Square Error}
\acrodef{PAC}[PAC]{Probably Approximately Correct}
\acrodef{PCA}[PCA]{Principal Component Analysis}
\acrodef{PDF}[PDF]{Probability Density Function}
\acrodef{PMF}[PMF]{Probability Mass Function}
\acrodef{PPM}[PPM]{Pulse Position Modulation}
\acrodef{PSD}{Power Spectral Density}
\acrodef{PSK}{Phase Shift Keying}
\acrodef{QKD}{Quantum Key Distribution}
\acrodef{ROC}{Receiver Operating Characteristic}
\acrodef{CVQKD}{Continuous-Variable \ac{QKD}}
\acrodef{QPSK}{Quadrature Phase-Shift Keying}
\acrodef{RV}{random variable}
\acrodef{SIMO}{Single-Input Multiple-Output}
\acrodef{SNR}{Signal-to-Noise Ratio}
\acrodef{SVM}[SVM]{Support Vector Machine}
\acrodef{POVM}{Positive Operator-Valued Measure}
\acrodefplural{POVM}{Positive Operator-Valued Measures}
\acrodef{wrt}[w.r.t.]{with respect to}
\acrodef{WSS}{Wide Sense Stationary}
\acrodef{RHS}{Right Hand Side}
\acrodef{LHS}{Left Hand Side}
\acrodef{CPTP}{Completely Positive and Trace Preserving}
\acrodef{ADSI}[ADSI]{Action-Dependent State Information}

\begin{document}

\title{Covert Communication and Key Generation Over Quantum State-Dependent Channels}

\author{
\IEEEauthorblockN{Hassan ZivariFard and R\'{e}mi A. Chou and Xiaodong Wang}\\
\thanks{H.~ZivariFard and X.~Wang are with the Department of Electrical Engineering, Columbia University, New York, NY 10027. R.~Chou is with the Department of Computer Science and Engineering, The University of Texas at Arlington, Arlington, TX 76019. The work of H.~ZivariFard and X.~Wang is supported in part by the U.S. Office of Naval Research (ONR) under grant N000142112155, and the work of R.~Chou is supported in part by NSF grant CCF-2047913. E-mails: \{hz2863, xw2008\}@columbia.edu and remi.chou@uta.edu. Part of this work was presented at the 2024 IEEE Information Theory Workshop.
}
}
\maketitle
\date{}
\begin{abstract}
We study covert communication and covert secret key generation with positive rates over quantum state-dependent channels. Specifically, we consider fully quantum state-dependent channels when the transmitter shares an entangled state with the channel. We study this problem setting under two security metrics. For the first security metric, the transmitter aims to communicate covertly with the receiver while simultaneously generating a covert secret key, and for the second security metric, the transmitter aims to transmit a secure message covertly and generate a covert secret key with the receiver simultaneously. Our main results include one-shot and asymptotic achievable positive covert-secret key rate pairs for both security metrics. Our results recover as a special case the best-known results for covert communication over state-dependent classical channels. To the best of our knowledge, our results are the first instance of achieving a positive rate for covert secret key generation and the first instance of achieving a positive covert rate over a quantum channel. Additionally, we show that our results are optimal when the channel is classical and the state is available non-causally at both the transmitter and the receiver. 
\end{abstract}

\section{Introduction}
\label{sec:Intro} 
Covert communication guarantees undetectable  communication~\cite{LPD_on_AWGN,Reliable_Deniable_Comm,LPD_by_Resolvability,LPD_over_DMC}. It is well known that in a
point-to-point classical \ac{DMC}, it is possible to reliably and covertly transmit at most on the order of $O(\sqrt{n})$  bits over $n$ channel uses \cite{LPD_by_Resolvability,LPD_over_DMC}, while the encoder and the decoder need to share a secret key on the order of $O(\sqrt{n})$ bits \cite{LPD_by_Resolvability}. However, positive covert communication rates are achievable for point-to-point classical \acp{DMC} \cite{LPD_over_DMC} when the symbol that the transmitter sends over the channel in the no-communication mode, referred to as $x_0\in\calX$, is redundant, i.e., it is possible to induce the same distribution as the distribution induced by $x_0\in\calX$ on the warden's channel observation with the channel input symbols $\{x\in\calX\}$. 
Motivated by this result, it has been shown that it is also possible to achieve positive covert rates over classical channels when a friendly jammer is present \cite{UninformedJammer,ISIT22,ISIT21,MyDissertation}, when the transmitter has access to \ac{CSI} \cite{Covert_With_State,Keyless22}, when the transmitter has access to  \ac{ADSI} \cite{ISIT23,Action_Covert}, when the warden has uncertainty about the statistical characterization of its channel~\cite{Lee15,Deniable_ITW14}, and when there is a cooperative user that knows the message or the transmitter's codeword \cite{ISIT22,MyDissertation,Action_Covert}. 

Existing works on covert communication and covert secret key generation over quantum channels show that optimal rates obey the square root law. Specifically, covert communication over bosonic channels is studied in \cite{LPD_on_bosonic,Wang23}, where the authors show that the covert capacity follows the square root law. Additionally, \cite{Sheikholeslami16,Wang16,Covert_Quantum16,Bullock_25} show that the covert capacity of classical-quantum point-to-point channels follows the square-root law. We note that even when the transmitter and the receiver are entangled, the covert capacity vanishes when the block length grows \cite{Entangled_Bosonic}. Covert quantum key distribution over quantum channels is studied in \cite{Covert_Secret_Key19,Covert_Secret_Key20,Covert_Secret_Key_Bos}, where the authors show that the capacity follows the square root law.

In this paper, unlike previous works, we show that it is possible to achieve positive covert rates and positive covert secret key generation rates over fully quantum state-dependent channels when the transmitter shares an entangled state with the channel. 
We study this problem under two different security metrics. For the first metric, the transmitter aims to communicate covertly with the receiver and simultaneously generate a covert secret key with the receiver. For the second metric, the transmitter aims to communicate both securely and covertly with the receiver and generate a covert secret key with the receiver, simultaneously. 
We present one-shot achievable covert rate regions and derive their counterparts in the asymptotic regime for both security metrics. 
The special case of our achievable rate regions for the classical asymptotic regime recovers the best achievable rates, derived in \cite{Covert_With_State} and \cite{Keyless22}, for covert communication over classical state-dependent channels. We also show that our results are optimal when the channel is classical and the state is available non-causally at both the transmitter and the receiver. To the best of our knowledge, our results are the first fundamental limits of covert communication with positive rates over quantum channels, and the first fundamental limits for covert secret key generation with positive rates over both classical and quantum channels.

The problem studied in this paper can also be considered as covert communication over a channel with a friendly jammer, where the transmitter has access to a state entangled with the quantum state that the jammer transmits over the channel. Our coding scheme in the one-shot regime is based on pinching \cite{QIT_Hayashi,Hayashi02}, through which we simultaneously perform channel resolvability to ensure covertness, and the Gel'fand-Pinsker type of encoding for aligning the codeword according to the \ac{CSI}. Compared with the achievability schemes in \cite{Dupuis09,DupuisDissertation} the main challenge is exploiting the \ac{CSI} for both message transmission and secret key generation and ensuring that our coding scheme is covert \ac{wrt} the warden, and compared to \cite{AshnuHayashi2020} the main challenge is exploiting the \ac{CSI} for both message transmission and secret key generation and inducing the distribution corresponding to the no communication mode on the warden's channel output observation. 

Some related works that do not consider the covertness constraint include secret key generation from classical \ac{CSI} \cite{SKA_with_State,SKA_with_State11} and secure communication over classical state-dependent channels \cite{WTC_With_State,Chia_WTC_with_State,Ziv_WTC_with_NonCausaul_CSI}. The problem of secure communication and secret key generation over state-dependent classical channels is studied in \cite{Bunin20,Han_WTC_with_CSI21}.   Communication over quantum channels with \ac{CSI} has first been studied in \cite{Dupuis09,DupuisDissertation}, and secure communication over quantum state-dependent channels has been studied in \cite{AshnuHayashi2020}. Some other work considered covert secret key generation over a state-dependent discrete memoryless classical channel with one-way public discussion in which the \ac{CSI} is not known at the transmitter and the warden may arbitrarily choose the channel state \cite{Covert_SKG20}, where the authors show that the covert capacity follows the square root law and provide inner and outer bounds on the covert capacity.

\textit{Notation:}
Let $\bbN^+$ be the set of positive natural numbers, $\bbR$ be the set of real numbers, and define $\bbR^+\triangleq\{x\in\bbR:x\geq 0\}$ and $\bbR^{++}\triangleq\mathbb{R}^+\backslash\{0\}$. We denote the set of positive semi-definite operators on a finite-dimensional Hilbert space $\calH$ by $\calP(\calH)$ and denote the set of quantum states by $\calD(\calH)\triangleq\{\rho\in\calP(\calH):\tra[\rho]=1\}$.  
We also denote the space of the bounded linear operators on $\calH$ with $\calL(\calH)$. For $\rho,\sigma\in\calD(\calH)$, the fidelity distance \cite{Jozsa94,Uhlmann76} between these two quantum states is denoted by $\F(\rho,\sigma)\triangleq\lVert\sqrt{\rho}\sqrt{\sigma}\rVert_1$, where $\lVert\rho\rVert_1\triangleq\tra\left[\sqrt{\rho^\dagger\rho}\right]$, and the purified distance \cite{Rastegin02,Tomamichel10} between these two quantum states is denoted by $\Pu(\rho,\sigma)\triangleq\sqrt{1-\F^2(\rho,\sigma)}$. For $\rho,\sigma\in\calD(\calH)$ and $\alpha\in(-1,0)\cup(0,\infty)$ the sandwiched R\'{e}nyi relative entropy \cite{Wilde14} is defined as $\ubar{\D}_{1+\alpha}\big(\rho\lVert\sigma\big)\triangleq\frac{1}{\alpha}\log\tra\left[\left(\sigma^{-\frac{\alpha}{2(1+\alpha)}}\rho\sigma^{-\frac{\alpha}{2(1+\alpha)}}\right)^{1+\alpha}\right]$ and the quantum relative entropy is defined as $\bbD(\rho\lVert\sigma)=\tra\sbr{\rho\pr{\log\rho-\log\sigma}}$. The identity operator  on some Hilbert space $\calH$ is denoted by $\bbI$. 
We use $\tau$ to denote a quantum density matrix and use $\ket{\tau}$ to denote the associated quantum state. For a finite set $\calX$, the relative entropy between  two distributions $P_X$ and $Q_X$ is denoted $\bbD(P_X\lVert Q_X)=\sum_{x\in\calX}P_X(x)\log\frac{P_X(x)}{Q_X(x)}$. For classical random variables, superscripts indicate the dimension of a vector, e.g., $X^n$, and $X_i^j$ denotes $(X_i,X_{i+1},\dots,X_j)$.

\section{Problem Statement}
\label{sec:Problem_Statement}
\subsection{One-Shot Regime} 
Fig.~\ref{fig:System_Model} illustrates a communication system designed to transmit classical information covertly over a quantum state-dependent channel when the \ac{CSI} is available at the transmitter in the sense that the transmitter shares an entangled state $\ket{\rho_{\bar{S}S}}$ with the channel. We define codes as follows. 
\begin{definition}
\label{defi:code}
A $(2^R,2^{R_K},1)$ code for the quantum state-dependent channel $\calN_{AS\to BE}$ when the \ac{CSI} is available at the encoder, consists of the following:
\begin{itemize}
    \item a message set $\calM\triangleq\sbra{1}{\left\lfloor{2^R}\right\rfloor}$ and a key set $\calK\triangleq\sbra{1}{\left\lfloor{2^{R_K}}\right\rfloor}$;
    \item an encoding map at the transmitter, which is a set of quantum channels $\calE\triangleq\{\calE^{(m)}_{\bar{S}\to AK}\}_{m\in\calM}$ that maps a state $\bar{S}\in\calD(\calH)$ to a channel input $A\in\calD(\calH)$ and a secret key $K\in\calK$;
    \item a decoding \ac{POVM} $\{\calD^{(m)}_{B\to MK}\}_{m\in\calM}$, which maps a channel observation $B\in\calD(\calH)$ to a message $\hat{M}\in\calM$ and a secret key $\hat{K}\in\calK$.
\end{itemize}
\end{definition}
The code is known by all the terminals, and the objective is to design a reliable and covert code. From Definition~\ref{defi:code}, the output of the legitimate receiver's channel is
\begin{subequations}\label{eq:Pe_rho0_CC_CSG}
\begin{align}
    B=\tra_E\calN_{AS\to BE}\tra_{K}\left(\calE^{(m)}_{\bar{S}\to AK}\otimes\bbI\left(\rho_{\bar{S}S}\right)\right).\label{eq:Legi_Output}
\end{align}Therefore, from Definition~\ref{defi:code} and \eqref{eq:Legi_Output} the probability of error is defined as
\begin{align}
    P_e&\triangleq\bbP\left\{(M,K)\ne(\hat{M},\hat{K})\right\}\nonumber\\
    &=\frac{1}{\card{\calM}\card{\calK}}\sum_{(m,k)\in\calM\times\calK}\tra\left[\left(\bbI-\calD^{(m)}_{B\to MK}\right)(B)\right].\label{eq:probaility_error}
\end{align}The code $(2^R,2^{R_K},1)$ is reliable if
\begin{align}
    P_e&\le\epsilon.\label{eq:Pe}
\end{align}
\end{subequations}When communication is not happening, the transmitter transmits the innocent state $\phi_0\in\calD(\calH)$ and therefore the warden receives the quantum state. 
\begin{align}
    \rho_0\triangleq\tra_B\calN_{AS\to BE}\left(\phi_0\otimes\tra_{\bar{S}}[\rho_{\bar{S}S}]\right).\label{eq:No_Comm}
\end{align}

Let $\rho_M$ denote the   state corresponding to the message $M$, which is 
    \begin{align}
        \rho_M\triangleq\frac{1}{2^R}\sum\limits_{m\in\sbra{1}{\left\lfloor{2^R}\right\rfloor}}\den{m}{m}.\label{eq:Message}
    \end{align}Also, since our objective is to generate a uniformly distributed secret key, we define the following   state for the secret key,
\begin{align}
    \sigma_K\triangleq\frac{1}{2^{R_K}}\sum\limits_{k\in\sbra{1}{\left\lfloor{2^{R_K}}\right\rfloor}}\den{k}{k}.\label{eq:Key_Uniformity}
\end{align}The joint state induced at the output of the warden by our code design is denoted by $\hat{\rho}_{MKE}$ and $\hat{\rho}_{KE}=\tra_M\left[\hat{\rho}_{MKE}\right]$. Then, we define the following two security metrics,
\begin{subequations}\label{eq:Security_Metrics}
    \begin{align}
        \left\lVert\hat{\rho}_{KE}-\sigma_K\otimes \rho_0\right\rVert_1&\le\delta,\label{eq:CC_CSK_Metric}\\
        \left\lVert\hat{\rho}_{MKE}-\rho_M\otimes\sigma_K\otimes \rho_0\right\rVert_1&\le\delta.\label{eq:CSC_CSK_Metric}
    \end{align}  \eqref{eq:CC_CSK_Metric} means that the generated secret key is almost uniform and independent of the output of the warden, and the state at the warden's output  approximates the state observed when no communication occurs, ensuring that  communication is covert.  \eqref{eq:CSC_CSK_Metric} means that the message is also independent of the generated secret key and the channel output of the warden, i.e., the message is secure.  
\end{subequations}

\begin{definition}[Covert Communication and Covert Secret Key Generation]
\label{defi:CC_SKG}
    A $(2^R,2^{R_K},1)$ code is $\epsilon$-reliable and $\delta$-covert if \eqref{eq:Pe} and \eqref{eq:CC_CSK_Metric} hold, respectively.
\end{definition}
\begin{definition}[Covert Secure Communication and Covert Secret Key Generation]
\label{defi:CSC_SKG}
A $(2^R,2^{R_K},1)$ code is $\epsilon$-reliable and $\delta$-covert secure, if \eqref{eq:Pe} and \eqref{eq:CSC_CSK_Metric} hold, respectively.
\end{definition}
\begin{figure}[t!]
\centering
\includegraphics[width=8.99cm]{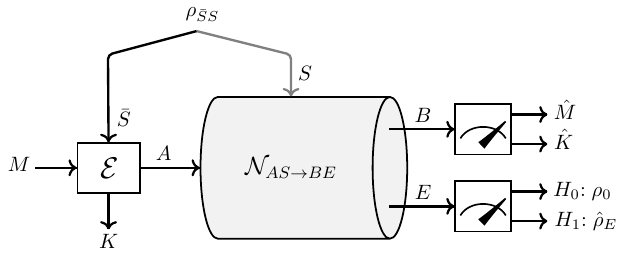}
\caption{Covert communication over a quantum state-dependent channel. The transmitter generates a secret key $K$ from its share of the \ac{CSI} $\bar{S}$, encodes the message $M$ along with the key $K$ into a quantum state $A$ by using the \ac{CSI} $\bar{S}$, and transmits it over a quantum channel $\calN_{AS\to BE}$. The goal is for the receiver to reliably decode the message and the secret key upon observing $B$, while the warden, observing $E$, should not be able to decide whether communication is taking place.}
\label{fig:System_Model}
\end{figure}
\subsection{Asymptotic Regime} 
We now define the problem in the asymptotic regime, where the channel is utilized $n$ times independently as $n$ approaches infinity. The corresponding $\left(2^{nR}, 2^{nR_K}, n\right)$ code for the channel $\calN_{AS\to BE}^{\otimes n}$ is similarly defined as Definition~\ref{defi:code}, with the message set $\calM\triangleq\sbra{1}{\left\lfloor{2^{nR}}\right\rfloor}$ and a key set $\calK\triangleq\sbra{1}{\left\lfloor{2^{nR_K}}\right\rfloor}$, quantum states $A^n, B^n, E^n\in\calD(\calH^n)$ and $\rho_{\bar{S}^nS^n}\in\calD(\calH^n\otimes\calH^n)$. Then, the output of the legitimate receiver's channel is
\begin{subequations}\label{eq:rho0_Pe_Asym}
\begin{align}
    B^n=\tra_{E^n}\calN_{A^nS^n\to B^nE^n}\tra_{K}\left(\calE^{(m)}_{\bar{S}^n\to A^nK}\otimes\bbI\left(\rho_{\bar{S}^nS^n}\right)\right),\label{eq:Legi_Output_Asymp}
\end{align}and the probability of error is
\begin{align} 
    P_e^{(n)}&\triangleq\bbP\left\{(M,K)\ne(\hat{M},\hat{K})\right\}\nonumber\\
    &=\frac{1}{\card{\calM}\card{\calK}}\sum_{(m,k)\in\calM\times\calK}\tra\left[\left(\bbI-\calD^{(m)}_{B^n\to MK}\right)\left(B^n\right)\right].\label{eq:Pe_Asym}
\end{align}A sequence of codes $(2^{nR},2^{nR_K},n)$ is reliable if
\begin{align}
    \lim_{n\to\infty}P_e^{(n)}&=0.\label{eq:Pe_Asymp}
\end{align}
\end{subequations}When communication is not happening, the transmitter transmits the innocent state $\phi_0^{\otimes n}\in\calD(\calH^n)$ and therefore the warden receives the following quantum state, 
\begin{subequations}\label{eq:Security_Metrics_Asymp}
\begin{align}
    \rho_0^{\otimes n}\triangleq\tra_{B^n}\calN_{AS\to BE}^{\otimes n}\left(\phi_0^{\otimes n}\otimes\tra_{\bar{S}}[\rho_{\bar{S}S}^{\otimes n}]\right).\label{eq:rho0_Asym}
\end{align}Note that since the \ac{CSI} is a tensor-product state, i.e., $\rho_{\bar{S}S}^{\otimes n}$, the state induced at the output of the warden, conditioned on the channel input being the innocent state $\phi_0^{\otimes n}$, is also a product state. 
Then, similar to \eqref{eq:Security_Metrics}, we define the following two security metrics,

    \begin{align}
        \lim_{n\to\infty}\left\lVert\hat{\rho}_{KE^n}-\sigma_K\otimes\rho_0^{\otimes n}\right\rVert_1&=0,\label{eq:CC_CSK_Metric_Asymp}\\
        \lim_{n\to\infty}\left\lVert\hat{\rho}_{MKE^n}-\rho_M\otimes\sigma_K\otimes\rho_0^{\otimes n}\right\rVert_1&=0,\label{eq:CSC_CSK_Metric_Asymp}
    \end{align}where $\rho_M$ and $\sigma_K$ are defined in \eqref{eq:Message} and \eqref{eq:Key_Uniformity}, respectively.
    \end{subequations}
\begin{definition}[Covert Communication and Covert Secret Key Generation]
\label{defi:Asymp_CC}
    A rate pair $(R,R_K)$ is said to be achievable for the quantum channel $\calN_{AS\to BE}$ if there exists a sequence of $\left(2^{nR},2^{nR_K},n\right)$ codes  such that \eqref{eq:Pe_Asymp} and \eqref{eq:CC_CSK_Metric_Asymp} hold.
\end{definition}
The covert capacity region, denoted by $\calC_{\textup{CC-CSK}}$,\footnote{CC stands for Covert Communication, and CSK stands for Covert Secret Key.} is defined as the set of all achievable covert message-key rate pairs. The covert capacity, denoted by $\C_{\textup{CC}}$, is defined as the supremum of all achievable covert message rates; and the covert secret key generation capacity, denoted by $\C_{\textup{CSK}}$, is defined as the supremum of all achievable covert secret key rates.
\begin{definition}[Covert Secure Communication and Covert Secret Key Generation]
\label{defi:Asymp_CSC_CSK}
    A rate pair $(R,R_K)$ is said to be achievable for the quantum channel $\calN_{AS\to BE}$ if there exists a sequence of $\left(2^{nR},2^{nR_K},n\right)$ codes  such that \eqref{eq:Pe_Asymp} and \eqref{eq:CSC_CSK_Metric_Asymp} hold.
\end{definition}
The covert secret capacity region, denoted by $\calC_{\textup{CSC-CSK}}$,\footnote{CSC stands for Covert and Secure Communication, and CSK stands for Covert Secret Key.} is defined as the set of all achievable covert secret message-secret key rate pairs. The covert secret capacity, denoted by $\C_{\textup{CSC}}$, is defined as the supremum of all achievable covert secret message rates. 
\subsection{Special Case: Classical Regime} 
A  special case of the problem illustrated in Fig.~\ref{fig:System_Model} corresponds to covert communication and covert secret key generation over classical state-dependent channels. Here the channel is $\big(\calS,\calA,W_{BE|AS},\calB,\calZ\big)$, where $\calS$ is the \ac{CSI} alphabet, $\calA$ is the channel input alphabet, $W_{BE|AS}$ is the channel, and $\calB$ and $\calZ$ are the channel output alphabets for the legitimate receiver and the warden, respectively. We assume that the channel state $S$ is known to the transmitter. Then, the stochastic encoder $\calE$ takes as input the message $M$ and the \ac{CSI} $S$, and outputs the channel input $A$ and the key $K$; and the stochastic decoder $\calD$ takes as input the channel output $B$ and outputs an estimation of the message and the secret key, i.e.,  $(\hat{M}, \hat{K})$.

In the classical case, $\rho_0$ in \eqref{eq:No_Comm} becomes the distribution induced at the warden's channel output when communication is not happening; $\hat{\rho}_{MKE}$ becomes the joint distribution induced by our code design; $\rho_M$ and $\sigma_K$ become uniform distributions. Then, for the security metrics defined in \eqref{eq:Security_Metrics} and \eqref{eq:Security_Metrics_Asymp}, the trace distance is replaced with the total variation distance.

Under the above problem setup, Definitions~\ref{defi:code} to \ref{defi:Asymp_CSC_CSK} have their counterparts for the classical channels. In particular, the classical covert capacity region $\calC_{\textup{CC-CSK}}$, and the classical covert secure capacity region, $\calC_{\textup{CSC-CSK}}$, can be similarly defined.

\section{Covert Communication, and Covert Secret Key Generation}
\label{sec:CC_CSK}
In this section, we first provide our one-shot result for the quantum state-dependent channels and then present our results for the asymptotic setup. Finally, we show that our results are optimal when the channel is classical and the state is available at both the transmitter and the~receiver.
\subsection{One-Shot Results} Our result on the existence of an $\epsilon$-reliable and $\delta$-covert $(2^R, 2^{R_K}, 1)$ code defined in Definition~\ref{defi:CC_SKG} is as follows.
\begin{theorem}[Covert Communication and Covert Secret Key Generation]
\label{thm:One_Shot}
Given a quantum state-dependent channel $\calN_{AS\to BE}$,  $\rho_{UAS}=\sum\limits_{u}p_{U}(u)\den{u}{u}_U\otimes\rho_{AS\lvert u}$, such that $\tra_{UA}[\rho_{UAS}]=\tra_{\bar{S}}[\rho_{\bar{S}S}]$, $\rho_E=\rho_0$, and $\rho_{UBE}\triangleq\bbI_U\otimes\calN_{AS\to BE}(\rho_{UAS})$, there exists a $(2^R,2^{R_K},1)$ code such that
\begin{subequations}\label{eq:PeCov_Lemma}
\begin{align}
    &\bbP\left\{(M,K)\ne(\hat{M},\hat{K})\right\}\nonumber\\
    &\le\frac{2v_S^\alpha}{\alpha}2^{\alpha\pr{-(R_J+R_K)+\ubar{\D}_{1+\alpha}\left(\rho_{US}\lVert\rho_U\otimes\rho_S\right)}}\nonumber\\
    &\qquad+12v_B^\alpha2^{\alpha\pr{R_J+R_K+R-\ubar{\D}_{1-\alpha}\big(\rho_{UB}\lVert\rho_U\otimes\rho_B\big)}},\label{eq:Pe_Lemma}\\
    &\left\lVert\hat{\rho}_{KE}-\sigma_K\otimes \rho_0\right\rVert_1\nonumber\\
    &\le\frac{2\sqrt{2}\pr{v_E}^{\frac{\alpha}{2}}}{\sqrt{\alpha}}2^{\frac{\alpha}{2}\pr{-(R_J+R)+\ubar{\D}_{1+\alpha}\left(\rho_{UE}\lVert\rho_U\otimes\rho_E\right)}},\label{eq:Covertnes_Lemma}
    \end{align}where $R_J>0$ is an arbitrary integer, $\alpha\in\left(0,\frac{1}{2}\right)$, and $v_S$, $v_B$, and $v_E$ are the number of distinct eigenvalues of the states $\rho_S$, $\rho_B$, and $\rho_E$, respectively.
\end{subequations}
\end{theorem}
The proof of Theorem~\ref{thm:One_Shot} is based on Gel'fand-Pinsker encoding for utilizing the \ac{CSI} for both communication and secret key generation, and channel resolvability for the covertness analysis. Compared to the coding scheme in \cite{GelfandPinsker,Dupuis09,AshnuHayashi2020} and inspired by \cite{Bunin20}, besides the bin index corresponding to the message and the bin index corresponding to the redundancy needed to correlate the codeword with the \ac{CSI}, our code introduces an additional redundancy bin index. Based on quantum state approximation arguments, we show that the new bin index is almost uniform and independent of the message and the warden's channel output. Since the legitimate receiver decodes all the bin indices of the transmitted codeword, it establishes the new redundancy bin index as a secret key. 
\begin{proof}
Let the distribution $p_U$ and the  states $\{\rho_{AS\lvert u}\}_{u\in\calU}$ be such that $\tra_{UA}[\rho_{UAS}]=\tra_{\bar{S}}[\rho_{\bar{S}S}]$. We define
\begin{subequations}
\begin{align}
    \rho_{UAS}&\triangleq \sum\limits_{u\in\calU}p_{U}(u)\den{u}{u}_U\otimes\rho_{AS\lvert u},\label{eq:Joint_State}\\
    \rho_{UBE}&\triangleq\bbI_U\otimes\calN_{AS\to BE}(\rho_{UAS}),\label{eq:Output_State}\\
    \rho_B&\triangleq\tra_{UE}[\rho_{UBE}],\label{eq:State_B}\\
    \rho_E&\triangleq\tra_{UB}[\rho_{UBE}].\label{eq:State_E}
\end{align}
\end{subequations}
\subsubsection*{Codebook Generation}Let 
$C\triangleq\big\{U(j,k,m)\big\}_{(j,k,m)\in\calJ\times\calK\times\calM}$, where $\calJ\triangleq\sbra{1}{\left\lfloor{2^{R_J}}\right\rfloor}$, $\calK\triangleq\sbra{1}{\left\lfloor{2^{R_K}}\right\rfloor}$, and $\calM\triangleq\sbra{1}{\left\lfloor{2^R}\right\rfloor}$, be a random codebook generated \ac{iid} according to $p_U$. Also, let $\calC\triangleq\big\{u(j,k,m)\big\}_{(j,k,m)\in\calJ\times\calK\times\calM}$ be a realization of the codebook $C$. The indices $j$, $k$, and $m$ can be seen as a two-layer random binning. For each $m\in\calM$, $C_m$ and $\calC_m$ denote a random sub-codebook and a deterministic sub-codebook of the codebooks $C$ and $\calC$, respectively, associated with the message $m$.

\subsubsection*{Encoding}
Hereafter, the quantum states induced by a fixed codebook $\calC$ are denoted by $\rho_{\cdot\lvert\calC}$, and the quantum states induced by a random codebook $C$ are denoted by $\rho_{\cdot\lvert C}$. 
We first introduce a state $Y$ such that $\ket{\rho_{YAS\lvert u}}$ is a purification of the state $\rho_{AS\lvert u}$, defined above. Therefore, we have the following purification of the state $\frac{1}{2^{R_J+R_K}}\sum_{(j,k)}\rho_{S\lvert u(j,k,m)}$, in which $\rho_{S\lvert u}=\tra_A[\rho_{AS\lvert u}]$,
\begin{align}
    \ket{\tau_{YASJK\lvert\calC_m}}\triangleq\frac{1}{\sqrt{2^{R_J+R_K}}}\sum_{(j,k)}\ket{\rho_{YAS\lvert u(j,k,m)}}\ket{j}_J\ket{k}_K.\nonumber
\end{align}From Uhlmann's theorem \cite{Uhlmann76}, there exists a set of isometries $\{W^{(\calC_m)}_{\bar{S}\to YAJK}\}$ such that
\begin{align}
    &\Pu\left(\tau_{YASJK\lvert\calC_m},W^{(\calC_m)}_{\bar{S}\to YAJK}\otimes\bbI_S\left(\rho_{\bar{S}S}\right)W^{\dagger(\calC_m)}_{\bar{S}\to YAJK}\otimes\bbI_S\right)=\nonumber\\
    &\qquad\Pu\left(\frac{1}{2^{R_J+R_K}}\sum_{(j,k)}\rho_{S\lvert u(j,k,m)},\rho_S\right).\label{eq:Isomery}
\end{align}
Our encoder depends on the codewords in the sub-codebook $\calC_m$ in the following manner. Given the message $m$, for each sub-codebook $\calC_m\triangleq\{u(j,k,m)\}_{(j,k)\in[1:2^{R_J}]\times[1:2^{R_K}]}$ the encoder applies an isometry $W^{(\calC_m)}_{\bar{S}\to YAJK}:\calD(\calH_{\bar{S}})\to\calD(\calH_Y\otimes\calH_J\otimes\calH_K\otimes\calH_A)$ to its share of the \ac{CSI}, i.e., the state $\rho_{\bar{S}}$. This isometry maps $\rho_{\bar{S}}$ to a purification state $Y$, the indices $J$ and $K$, and the channel input $A$. The encoder then transmits $A$ over the channel $\calN_{AS\to BE}$.
\subsubsection*{Pinching}
\label{sec:Pinching}
Our decoding procedure is based on the pinching method \cite{QIT_Hayashi}. For any state $\sigma\in\calD(\calH)$ and $\rho\in\calD(\calH)$, let $\sigma=\sum_{i}\lambda_i\den{z_i}{z_i}$ be the spectral decomposition of the state $\sigma$.  The pinching operation of the state $\rho$ \ac{wrt} the spectral decomposition of the state $\sigma$ is defined as $\calE_\sigma(\rho)\triangleq\sum_i\den{z_i}{z_i}\rho\den{z_i}{z_i}=\sum_i\bra{z_i}\rho\ket{z_i}\den{z_i}{z_i}$. Note that the state $\sigma$ and the state $\calE_\sigma(\rho)$ commute. 

We aim to find pinching operations such that after applying the pinching operation, the state $\rho_{UB}$ commutes with the state $\rho_U\otimes\rho_B$ and the state $\rho_{US}$ commutes with $\rho_U\otimes\rho_S$. 
Let $\calE_{\rho_B}$ be the pinching operation \ac{wrt} the spectral decomposition of the state $\rho_B$. Therefore, $\calE_{\rho_B}(\rho_{UB})$ and the state $\rho_U\otimes\rho_B$ commute. 
Also, let $\calE_{\rho_S}$ be the pinching operation \ac{wrt} the spectral decomposition of the state $\rho_S$. Therefore, $\calE_{\rho_S}(\rho_{US})$ and the state $\rho_S$ commute. We denote the number of distinct eigenvalues of the states $\rho_B$ and $\rho_S$ with $v_B$ and $v_S$, respectively.

\begin{figure*}[t!]
\centering
\includegraphics[width=10.0cm]{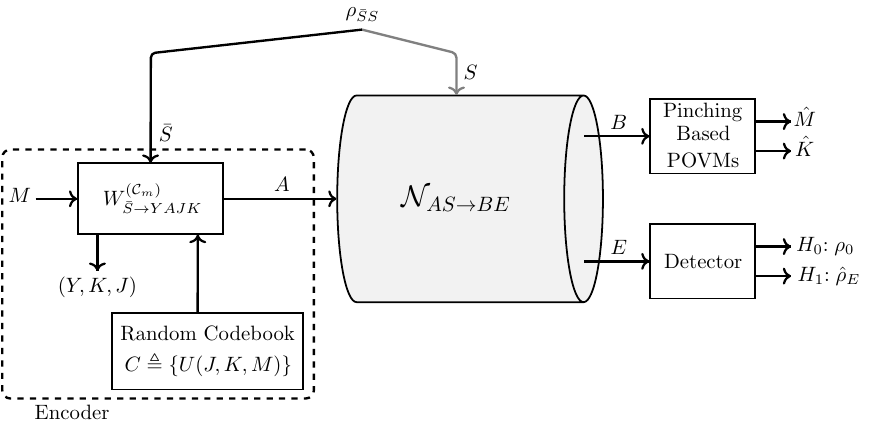}
\caption{The encoder generates a random codebook $C\triangleq\{U(j,k,m)\}_{(j,k,m)\in\calJ\times\calK\times\calM}$, where $m\in\calM\triangleq\sbr{1:2^R}$ is the message, $k\in\calK\triangleq\sbr{1:2^{R_K}}$ is the secret key, and $j\in\calJ\triangleq\sbr{1:2^{R_J}}$ represents the local randomness needed to perform Gel'fand-Pinsker encoding and also represents part of the randomness needed to secure the key $k$. Since the message $m$ is not required to be secure, it also serves as part of the local randomness needed to secure the key. To transmit the message $m$, given the \ac{CSI} $\bar{S}$ and the fixed sub-codebook $\calC_m\triangleq\{U(j,k,m)\}_{(j,k)\in\calJ\times\calK}$, the encoder applies the isometry $W^{(\calC_m)}_{\bar{S}\to YAJK}$ on its share of the \ac{CSI}, i.e., $\bar{S}$, which outputs the purification state $Y$, the index $J$, the secret key $K$, and the channel input $A$. Then it transmits the quantum state $A$ over the channel. The decoder applies a set of pinching-based \ac{POVM} operators on its received state $B$ to recover the message $m$ and the secret key $ K$. Using the channel resolvability techniques and properties of the pinching maps, we bound the probability of detection at the warden.}
\label{fig:Coding_Scheme}
\end{figure*}
\subsubsection*{Decoding}
\label{sec:Decoding}
For any two Hermitian matrices $A$ and $B$, we define the projection $\{A\ge B\}$ as $\sum_{i:\lambda_i\ge0}P_i$, where the spectral decomposition of $A-B$ is $\sum_i\lambda_iP_i$ and $P_i$ is the projection to the eigenspace corresponding to the eigenvalue $\lambda_i$. Now, we define the following projectors,
\begin{align}
    \Pi_{UB}&\triangleq\left\{\calE_{\rho_B}(\rho_{UB})\ge2^{R_J+R_K+R}\rho_U\otimes\rho_B\right\}.\label{eq:Main_Projection}
\end{align}
Also, for every $j\in\sbra{1}{2^{R_J}}$, $k\in\sbra{1}{2^{R_K}}$, and $m\in\sbra{1}{2^R}$ we define the following operator,
\begin{align}
    \Gamma(j,k,m)\triangleq\tra_U\left[\Pi_{UB}\left(\den{U(j,k,m)}{U(j,k,m)}\otimes\bbI_B\right)\right].\label{eq:Operator}
\end{align}To obtain a set of \ac{POVM} operators, we normalize \eqref{eq:Operator} as follows,
\begin{align}
    \Upsilon(j,k,m)&\triangleq\left(\sum_{j',k',m'}\Gamma(j',k',m')\right)^{-\frac{1}{2}}\nonumber\\
    &\quad\times\Gamma(j,k,m)\left(\sum_{j',k',m'}\Gamma(j',k',m')\right)^{-\frac{1}{2}}.\label{eq:POVM_Decoder}
\end{align}The receiver decodes the secret key and the message by using the  \ac{POVM} operator in \eqref{eq:POVM_Decoder}. The coding scheme for Theorem~\ref{thm:One_Shot} is summarized in Fig.~\ref{fig:Coding_Scheme}.  
The following lemmas are essential for the error analysis and the covertness analysis.
\begin{lemma}
    \label{lemma:Resolvability_1}
    Let $\rho_{UE}\triangleq\sum_{u\in\calU}p_U(u)\den{u}{u}\otimes\rho_{E\lvert u}$ be a classical-quantum state. Also, let  $C_U\triangleq\{U(i)\}_{i\in\calI}$, where $\calI\triangleq\left[2^R\right]$, be a set of random variables in which $U(i)$ is generated \ac{iid} according to $p_U$, and define $\tau_{E\lvert C}\triangleq\frac{1}{2^R}\sum_i\rho_{E\lvert U(i)}$. Then, for $\alpha\in\left(0\,,\frac{1}{2}\right)$, 
    \begin{align}
        \bbE_C\left[\ubar{\D}_{1+\alpha}\big(\tau_{E\lvert C}\lVert\rho_E\big)\right]
        &\le\frac{v_E^\alpha}{\alpha\ln2}2^{\alpha\left(-R+\ubar{\D}_{1+\alpha}\left(\rho_{UE}\lVert\rho_U\otimes\rho_E\right)\right)},
    \end{align}where $v_E$ is the number of distinct eigenvalues of the state $\rho_E$.
\end{lemma}
\begin{proof}
   See  Appendix~\ref{proof:lemma:Resolvability_1}.
\end{proof}
\begin{lemma}
\label{lemma:Decodability}
    For $\alpha\in\left(0\,,\frac{1}{2}\right)$ we have,
    \begin{subequations}\label{eq:Decoding_Lemma}
    \begin{align}
    &\tra\left[\left(\bbI-\Pi_{UB}\right)\rho_{UB}\right]\nonumber\\
    &\le v_B^\alpha2^{\alpha(R_J+R_K+R)}2^{-\alpha\ubar{\D}_{1-\alpha}\big(\rho_{UB}\lVert\rho_U\otimes\rho_B\big)},\label{eq:Decoding_Lemma_1}\\
    &\tra\left[\Pi_{UB}\left(\rho_U\otimes\rho_B\right)\right]\nonumber\\
    &\le v_B^\alpha2^{-(1-\alpha)(R_J+R_K+R)}2^{-\alpha\ubar{\D}_{1-\alpha}\big(\rho_{UB}\lVert\rho_U\otimes\rho_B\big)}.\label{eq:Decoding_Lemma_2}
    \end{align}
\end{subequations}
\end{lemma}
\begin{proof}
   See  Appendix~\ref{proof:lemma:Decodability}.
\end{proof}
\subsubsection*{Error Analysis}
\label{sec:error_analysis}
\sloppy Now let $\Theta_B(1)$ be the state the legitimate receiver receives when the transmitter transmits the message $m=1$. Therefore, 
\begin{align}
    &\Theta_B(1)=\tra_E\calN_{AS\to BE}\nonumber\\
    &\tra_{YJK}\left(W^{(\calC_1)}_{\bar{S}\to YAJK}\otimes\bbI_S\left(\rho_{\bar{S}S}\right)W^{\dagger(\calC_1)}_{\bar{S}\to YAJK}\otimes\bbI_S\right).\nonumber
\end{align}
Also, let
\begin{align}
    \hat{\Theta}_B(1)&\triangleq\frac{1}{2^{R_J+R_K}}\sum_{j\in\sbra{1}{2^{R_J}}}\sum_{k\in\sbra{1}{2^{R_K}}}\nonumber\\
    &\quad\tra_E\calN_{AS\to BE}\tra_Y\left(\rho_{YAS\lvert U(j,k,1)}\right).\nonumber
\end{align}To bound the probability of error average over the random choice of the codebook, from the symmetry of the codebook construction, it is sufficient to bound $\bbE_C\bbP\left\{\hat{M}\ne1\lvert M=1\right\}$ as follows,
\begin{align}
    &\bbE_C\bbP\left\{\hat{M}\ne1\lvert M=1\right\}\nonumber\\
    &\quad=\bbE_C\left[\tra\left[\left(\sum_{j,k,m'\ne1}\Upsilon(j,k,m')\right)\Theta_B(1)\right]\right]\nonumber\\
    &\quad\mathop\le\limits^{(a)}2\bbE_C\left[\tra\left[\left(\sum_{j,k,m'\ne1}\Upsilon(j,k,m')\right)\hat{\Theta}_B(1)\right]\right]\nonumber\\
    &\quad\qquad+2\bbE_C\left|\sqrt{\tra\left[\left(\sum_{j,k,m'\ne1}\Upsilon(j,k,m')\right)\Theta_B(1)\right]}\right.\nonumber\\
    &\quad\qquad\left.-\sqrt{\tra\left[\left(\sum_{j,k,m'\ne1}\Upsilon(j,k,m')\right)\hat{\Theta}_B(1)\right]}\right|^2\nonumber\\
    &\quad\mathop\le\limits^{(b)}2\bbE_C\left[\tra\left[\left(\sum_{j,k,m'\ne1}\Upsilon(j,k,m')\right)\hat{\Theta}_B(1)\right]\right]\nonumber\\
    &\quad\qquad+2\bbE_C\left[\Pu^2\left(\hat{\Theta}_B(1),\Theta_B(1)\right)\right]\nonumber\\
    &\quad\mathop=\limits^{(c)}2\sum_{j,k,m'\ne1}\bbE_C\left[\tra\left[\Upsilon(j,k,m')\hat{\Theta}_B(1)\right]\right]\nonumber\\
    &\quad\qquad+2\bbE_C\left[\Pu^2\left(\hat{\Theta}_B(1),\Theta_B(1)\right)\right],\label{eq:Error_Prob_m1}
\end{align}
\begin{itemize}
    \item[$(a)$] follows since $(x+y)^2\le2(x^2+y^2)$;
    \item[$(b)$] follows since $\sum_{j,k,m'\ne1}\Upsilon(j,k,m')\preceq\bbI$ and from  \cite{Anurag19} for two states $\rho_A\in\calD(\calH)$ and $\sigma_A\in\calD(\calH)$ and for every operator $0\preceq\Lambda\preceq\bbI$ we have
    \begin{align}
        \abs{\sqrt{\tra[\Lambda\rho]}-\sqrt{\tra[\Lambda\sigma]}}\le\Pu(\rho,\sigma);\nonumber
    \end{align}
    \item[$(c)$] follows since the trace operation and the expectation are linear. 
\end{itemize}We first bound the second term on the \ac{RHS} of \eqref{eq:Error_Prob_m1} as follows,
\begin{align}
    &2\bbE_C\left[\Pu^2\left(\hat{\Theta}_B(1),\Theta_B(1)\right)\right]\nonumber\\
    &=2\bbE_C\Big[\Pu^2\Big(\frac{1}{2^{R_J+R_K}}\sum_{j\in\sbra{1}{2^{R_J}}}\sum_{k\in\sbra{1}{2^{R_K}}}\tra_{YE}\nonumber\\
    &\quad\calN_{AS\to BE}\left(\rho_{YAS\lvert U(j,k,1)}\right),\tra_{YEJK}\nonumber\\
    &\quad\calN_{AS\to BE}\left(W^{(\calC_1)}_{\bar{S}\to YAJK}\otimes\bbI_S\left(\rho_{\bar{S}S}\right)W^{\dagger(\calC_1)}_{\bar{S}\to YAJK}\otimes\bbI_S\right)\Big)\Big]\nonumber\\
    &\mathop\le\limits^{(a)}2\bbE_C\left[\Pu^2\left(\tau_{YASJK\lvert C}\right.\right.\nonumber\\
    &\left.\left.,W^{(\calC_1)}_{\bar{S}\to YAJK}\otimes\bbI_S\left(\rho_{\bar{S}S}\right)W^{\dagger(\calC_1)}_{\bar{S}\to YAJK}\otimes\bbI_S\right)\right]\nonumber\\
    &\mathop=\limits^{(b)}2\bbE_C\left[\Pu^2\left(\tau_{S\lvert C},\rho_S\right)\right]\nonumber\\
    &\mathop\le\limits^{(c)}2\bbE_C\left[1-2^{-\ubar{\D}_{1+\alpha}\left(\tau_{S\lvert C}\lVert\rho_S\right)}\right]\nonumber\\
    &\mathop\le\limits^{(d)}2\ln2\bbE_C\left[\ubar{\D}_{1+\alpha}\left(\tau_{S\lvert C}\lVert\rho_S\right)\right]\nonumber\\
    &\mathop\le\limits^{(e)}\frac{2v_S^\alpha}{\alpha}2^{\alpha\left(-(R_J+R_K)+\ubar{\D}_{1+\alpha}\left(\rho_{US}\lVert\rho_U\otimes\rho_S\right)\right)},\label{eq:Distance_Two_dist}
\end{align}where
\begin{itemize}
    \item[$(a)$] follows since from \cite{Branum96,Frank13} for two  states $\rho_A\in\calD(\calH)$ and $\sigma_A\in\calD(\calH)$ and quantum channel $\calN(\cdot):\calL(X)\to\calL(Y)$ we have 
    \begin{align}
        \Pu\big(\calN(\rho),\calN(\sigma)\big)\le\Pu(\rho,\sigma);\nonumber
    \end{align}
    \item[$(b)$] follows from \eqref{eq:Isomery} with $\tau_{S\lvert C}\triangleq\frac{1}{2^{R_J+R_K}}\sum_{(j,k)}\rho_{S\lvert U(j,k,1)}$;
    \item[$(c)$] follows since for two states $\rho\in\calD(\calH)$ and $\sigma\in\calD(\calH)$ we have \cite[Corollary~4.3]{Tomamichel16},
    \begin{align}
        F^2(\rho,\sigma)=2^{-\ubar{\D}_{1/2}\left(\rho\lVert\sigma\right)}\ge2^{-\ubar{\D}_{1+\alpha}\left(\rho\lVert\sigma\right)};\label{eq:purified-sandReny}
    \end{align}
    \item[$(d)$] follows since 
    \begin{align}
        1-2^{-\frac{x}{\ln2}}=1-e^{-x}\le x;\label{eq:Usefull_Inequality}
    \end{align}
    \item[$(e)$] follows Lemma~\ref{lemma:Resolvability_1}.
\end{itemize}
We now bound the first term on the \ac{RHS} of \eqref{eq:Error_Prob_m1} as follows,
\begin{align}
    &2\sum_{j,k,m'\ne1}\bbE_C\left[\tra\left[\Upsilon(j,k,m')\hat{\Theta}_B(1)\right]\right]\nonumber\\
    &=\frac{2}{2^{R_J+R_K}}\sum_{(j,k)}\sum_{(j',k',m'\ne1)}\bbE_C\left[\tra\left[\Upsilon(j',k',m')\rho_{B\lvert U(j,k,1)}\right]\right]\nonumber\\
    &\mathop=\limits^{(a)}2\sum_{(j',k',m'\ne1)}\bbE_C\left[\tra\left[\Upsilon(j',k',m')\rho_{B\lvert U(1,1,1)}\right]\right]\nonumber\\
    &\le2\bbE_C\left[\tra\left[\left(\bbI-\Upsilon(1,1,1)\right)\rho_{B\lvert U(1,1,1)}\right]\right]\nonumber\\
    &\mathop\le\limits^{(b)}4\bbE_C\left[\tra\left[\left(\bbI-\Gamma(1,1,1)\right)\rho_{B\lvert U(1,1,1)}\right]\right]\nonumber\\
    &+8\sum_{(j,k,m')\ne(1,1,1)}\bbE_C\left[\tra\left[\Gamma(j,k,m')\rho_{B\lvert U(1,1,1)}\right]\right],\label{eq:Reliability_Second_Term}
\end{align}where $(a)$ follows from the symmetry of the codebook construction \ac{wrt} $(j,k)$, and $(b)$ follows from Hayashi-Nagaoka inequality \cite{Hayashi03}. We now bound the first term on the \ac{RHS} of \eqref{eq:Reliability_Second_Term} as 
\begin{align}
    &4\bbE_C\left[\tra\left[\left(\bbI-\Gamma(1,1,1)\right)\rho_{B\lvert U(1,1,1)}\right]\right]\nonumber\\
    &\mathop=\limits^{(a)}4\bbE_C\left[\tra\left[\left(\bbI-\tra_U\left[\Pi_{UB}\left(\den{U(1,1,1)}{U(1,1,1)}\otimes\bbI_B\right)\right]\right)\right.\right.\nonumber\\
    &\qquad\left.\left.\times\rho_{B\lvert U(1,1,1)}\right]\right]\nonumber\\
    &\mathop=\limits^{(b)}4\tra\left[\left(\bbI-\Pi_{UB}\right)\rho_{UB}\right]\nonumber\\
    &\mathop\le\limits^{(c)}4v_B^\alpha2^{\alpha(R_J+R_K+R)}2^{-\alpha\ubar{\D}_{1-\alpha}\big(\rho_{UB}\lVert\rho_U\otimes\rho_B\big)},\label{eq:First_Term_Pe}
\end{align}where
\begin{itemize}
    \item[$(a)$] follows from the definition of $\Gamma(1,1,1)$ in \eqref{eq:Operator};
    \item[$(b)$] follows from the linearity of the trace operation and the expectation, and taking the expectation \ac{wrt} the random codebook $C$;
    \item[$(c)$] follows from Lemma~\ref{lemma:Decodability}.
\end{itemize}
\begin{figure*}[b!]
\hrulefill
\setcounter{equation}{23}
\begin{align}
    &\bbE_C\Pu^2\left(\frac{1}{2^{R_J+R_K+R}}\sum_{(j,k,m)}\den{k}{k}\otimes\rho_{E\lvert U(j,k,m)},\frac{1}{2^{R_K}}\sum_{k\in\left[1:2^{R_K}\right]}\den{k}{k}\otimes\rho_E\right)\nonumber\\
    &=1-\bbE_CF^2\left(\frac{1}{2^{R_J+R_K+R}}\sum_{(j,k,m)}\den{k}{k}\otimes\rho_{E\lvert U(j,k,m)},\frac{1}{2^{R_K}}\sum_{k\in\left[1:2^{R_K}\right]}\den{k}{k}\otimes\rho_E\right)\nonumber\\
    &\mathop\le\limits^{(a)}2-2\bbE_CF\left(\frac{1}{2^{R_J+R_K+R}}\sum_{(j,k,m)}\den{k}{k}\otimes\rho_{E\lvert U(j,k,m)},\frac{1}{2^{R_K}}\sum_{k\in\left[1:2^{R_K}\right]}\den{k}{k}\otimes\rho_E\right)\nonumber\\
    &\mathop=\limits^{(b)}2-2\bbE_C\left(\frac{1}{2^{R_K}}\sum_{k\in\left[1:2^{R_K}\right]}F\left(\frac{1}{2^{R_J+R}}\sum_{(j,m)}\rho_{E\lvert U(j,k,m)},\rho_E\right)\right)\nonumber\\
    &\mathop\le\limits^{(c)}\frac{2}{2^{R_K}}\sum_{k\in\left[1:2^{R_K}\right]}\bbE_C\Pu^2\left(\frac{1}{2^{R_J+R}}\sum_{(j,m)}\rho_{E\lvert U(j,k,m)},\rho_E\right)\nonumber\\
    &\mathop\le\limits^{(d)}\frac{2}{2^{R_K}}\sum_{k\in\left[1:2^{R_K}\right]}\bbE_C\left[1-2^{-\ubar{\D}_{1+\alpha}\left(\frac{1}{2^{R_J+R}}\sum_{(j,m)}\rho_{E\lvert U(j,k,m)}\big\lVert\rho_E\right)}\right]\nonumber\\
    &\mathop\le\limits^{(e)}\frac{2\ln2}{2^{R_K}}\sum_{k\in\left[1:2^{R_K}\right]}\bbE_C\left[\ubar{\D}_{1+\alpha}\left(\frac{1}{2^{R_J+R}}\sum_{(j,m)}\rho_{E\lvert U(j,k,m)}\Big\lVert\rho_E\right)\right]\nonumber\\
    &\mathop\le\limits^{(f)}\frac{2v_E^\alpha}{\alpha}2^{\alpha\left(-(R_J+R)+\ubar{\D}_{1+\alpha}\left(\rho_{UE}\lVert\rho_U\otimes\rho_E\right)\right)},\label{eq:Final_Covertness}
\end{align}
\setcounter{equation}{21}
\end{figure*}
We now bound the second term on the \ac{RHS} of \eqref{eq:Reliability_Second_Term} as
\begin{align}
    &8\sum_{(j,k,m')\ne(1,1,1)}\bbE_C\left[\tra\left[\Gamma(j,k,m')\rho_{B\lvert U(1,1,1)}\right]\right]\nonumber\\
    &\quad\mathop=\limits^{(a)}8\sum_{(j,k,m')\ne(1,1,1)}\bbE_C\left[\tra\left[\tra_U\left[\Pi_{UB}\right.\right.\right.\nonumber\\
    &\qquad\left.\left.\left.\times\left(\den{U(j,k,m')}{U(j,k,m')}\otimes\bbI_B\right)\right]\rho_{B\lvert U(1,1,1)}\right]\right]\nonumber\\
    &\quad=8\sum_{(j,k,m')\ne(1,1,1)}\tra\left[\bbE_C\left[\Pi_{UB}\right.\right.\nonumber\\
    &\qquad\,\left.\left.\times\left(\den{U(j,k,m')}{U(j,k,m')}\otimes\rho_{B\lvert U(1,1,1)}\right)\right]\right]\nonumber\\
    &\quad\mathop=\limits^{(b)}8\sum_{(j,k,m')\ne(1,1,1)}\tra\left[\Pi_{UB}\left(\rho_U\otimes\rho_B\right)\right]\nonumber\\
    &\quad\le8\times2^{(R_J+R_K+R)}\tra\left[\Pi_{UB}\left(\rho_U\otimes\rho_B\right)\right]\nonumber\\
    &\quad\mathop\le\limits^{(c)}8v_B^\alpha2^{\alpha(R_J+R_K+R)}2^{-\alpha\ubar{\D}_{1-\alpha}\big(\rho_{UB}\lVert\rho_U\otimes\rho_B\big)},\label{eq:Second_Term_Pe}
\end{align}where
\begin{itemize}
    \item[$(a)$] follows from the definition of $\Gamma(j,k,m')$ in \eqref{eq:Operator};
    \item[$(b)$] follows from the linearity of the trace operation and the expectation and taking the expectation \ac{wrt} the random codebook $C$ and the fact that the symbols $U(i,j,k)$ is independent of $U(i',j',k')$, for $(i,j,k)\ne (i',j',k')$; 
    \item[$(c)$] follows from Lemma~\ref{lemma:Decodability}.
\end{itemize}
Then, substituting \eqref{eq:First_Term_Pe} and \eqref{eq:Second_Term_Pe} into \eqref{eq:Reliability_Second_Term} and substituting \eqref{eq:Distance_Two_dist} and \eqref{eq:Reliability_Second_Term} back to \eqref{eq:Error_Prob_m1} leads to,
\begin{align}
    &\bbE_C\bbP\left\{\hat{M}\ne1\lvert M=1\right\}\nonumber\\
    &\le\frac{2v_S^\alpha}{\alpha}2^{\alpha\left(-(R_J+R_K)+\ubar{\D}_{1+\alpha}\left(\rho_{US}\lVert\rho_U\otimes\rho_S\right)\right)}\nonumber\\
    &\qquad+12v_B^\alpha2^{\alpha\left(R_J+R_K+R-\ubar{\D}_{1-\alpha}\big(\rho_{UB}\lVert\rho_U\otimes\rho_B\big)\right)}.\label{eq:Final_Pe}
\end{align}
\subsubsection*{Covertness and Security Analysis}
 To prove the covertness of our code design, we bound $\bbE_C\Pu^2\left(\hat{\rho}_{KE},\sigma_K\otimes\rho_E\right)$, where $\hat{\rho}_{KE}\triangleq\frac{1}{2^{R_J+R_K+R}}\sum_{(j,k,m)}\den{k}{k}\otimes\rho_{E\lvert U(j,k,m)}$ is the state induced by our code design, and then choose $\rho_{UAS}$ such that $\rho_E=\rho_0$. We have, \eqref{eq:Final_Covertness} at the bottom of the page, where
\begin{itemize}
\setcounter{equation}{24}
    \item[$(a)$] follows since
    \begin{align}
        1-F^2(\rho,\sigma)&=\big(1+F(\rho,\sigma)\big)\big(1-F(\rho,\sigma)\big)\nonumber\\
        &\le2\big(1-F(\rho,\sigma)\big),\label{eq:Purified_Bure}
    \end{align}where the  inequality follows since $0\le\F(\rho,\sigma)\le1$;
    \item[$(b)$] follows from the direct-sum property of the fidelity distance \cite[Eq.~(3.8)]{Fidelity24},
    \begin{align}
        &\F\left(\sum_{x\in\calX}p(x)\den{x}{x}\otimes\rho_x,\sum_{x\in\calX}p(x)\den{x}{x}\otimes\sigma_x\right)\nonumber\\
        &=\sum_{x\in\calX}p(x)\F\left(\rho_x,\sigma_x\right);\label{eq:Direct_Sum_Fidelity}
    \end{align} 
    \item[$(c)$] follows since $0\le\F(\rho,\sigma)\le1$, we have $1-F(\rho,\sigma)\le\Pu^2(\rho,\sigma)$;
    \item[$(d)$] follows from \eqref{eq:purified-sandReny} and the definition of the purified distance;
    \item[$(e)$] follows from \eqref{eq:Usefull_Inequality};
    \item[$(f)$] follows from Lemma~\ref{lemma:Resolvability_1}.
\end{itemize}Then, for two arbitrary states $\rho\in\calD(\calH)$ and $\sigma\in\calD(\calH)$ we have \cite[Theorem~9.3.1]{Wilde_Book}
\begin{align}
    \lVert\rho-\sigma\rVert_1\le2\Pu(\rho,\sigma).\label{eq:Tr_Dist_Pu_Dist}
\end{align}Therefore, from  \eqref{eq:Tr_Dist_Pu_Dist} and \eqref{eq:Final_Covertness}, we obtain \eqref{eq:Covertnes_Lemma}.
\end{proof}

\subsection{Asymptotic Results}The following result can be deduced from Theorem~\ref{thm:One_Shot}, which establishes an achievable rate region for the asymptotic regime defined in Definition~\ref{defi:Asymp_CC}.
\begin{theorem}
\label{thm:Asymp}
Define the rate region $\calR_{\textup{CC-CSK}}$ as follows,
\begin{align}%
&\calR_{\textup{CC-CSK}}\nonumber\\
&\triangleq
\bigcup_{\rho_{UAS}} %
\left\{ \begin{array}{rl}
  (R,R_K) \,:\;
	R &\leq I(U;B)- I(U;S)\\
        R_K &\leq I(U;B)- I(U;E)\\
	R+R_K &\leq I(U;B)
	\end{array}
\right\},
\label{eq:inRnone}
\end{align}where $\rho_{UAS}=\sum\limits_{u}p_{U}(u)\den{u}{u}_U\otimes\rho_{AS\lvert u}$ such that $\tra_{UA}[\rho_{UAS}]=\tra_{\bar{S}}[\rho_{\bar{S}S}]$ and $\rho_E=\rho_0$. 
An inner bound for the covert capacity region of a quantum state-dependent channel $\calN_{AS\to BE}$, depicted in Fig.~\ref{fig:System_Model}, is
    \begin{align}
        \calC_{\textup{CC-CSK}}\supseteq\calR_{\textup{CC-CSK}}.\nonumber
    \end{align}
\end{theorem}The proof of Theorem~\ref{thm:Asymp} is based on Theorem~\ref{thm:One_Shot} assuming that the transmitter is allowed to use the channel $n$ times independently, when $n$ goes to infinity. The penalty term $I(U;S)$ represents the reduction in transmission rate required for the redundancy we need to add to correlate the codeword with the channel state, while the penalty term $I(U;E)$ accounts for the rate loss required to generate a covert \textit{secure} key.
\begin{proof}
We first bound the number of distinct components of the pinching maps $\calE_{\rho_S}$ and $\calE_{\rho_E}$ in the asymptotic and \ac{iid} case.   
Assuming that the dimension of the Hilbert space of $\rho_B$ is $d_B$, $\rho_B$ at most has $d_B$ eigenvalues. From \cite[Lemma~2.2]{Csiszar_Korner_Book} the number of different types of sequences $X^n\in\calX^n$ is less than $(n+1)^{\card{\calX}}$, therefore, $\rho_B^{\otimes n}$ has at most $(n+1)^{d_B}$ eigenvalues. Similarly, assuming that the dimension of the Hilbert space of $\rho_E$ is $d_E$, we can bound the number of eigenvalues of $\rho_E$, and therefore we have
\begin{align}
    v_S\le(n+1)^{d_S},\quad v_E\le(n+1)^{d_E}.\label{eq:vS_vB}
\end{align} 
Then, from \eqref{eq:Final_Pe} and \eqref{eq:Final_Covertness}, assuming that the $U^n$ codewords are generated in an \ac{iid} manner, the corresponding classical-quantum state is a product state, i.e.,
\begin{align}
    \rho_U^{\otimes n}&=\sum_{(u_1,u_2,\cdots,u_n)}\prod_{t=1}^nP_U(u_t)\den{u_1}{u_1}\otimes\den{u_2}{u_2}\otimes\cdots\nonumber\\
    &\qquad\otimes\den{u_n}{u_n}.\nonumber
\end{align}When the channel $\calN_{AS \to BE}$ is used $n$ times independently, there exists a code such that
\begin{subequations}\label{eq:Final_PeCov_Asymp}
\begin{align}
    &\bbP\left\{M\ne\hat{M}\right\}\le\frac{2v_S^\alpha}{\alpha}2^{\alpha\left(-n(R_J+R_K)+\ubar{\D}_{1+\alpha}\left(\rho_{US}^{\otimes n}\lVert\rho_U^{\otimes n}\otimes\rho_S^{\otimes n}\right)\right)}\nonumber\\
    &\quad+12v_B^\alpha2^{\alpha\left(n(R_J+R_K+R)-\ubar{\D}_{1-\alpha}\big(\rho_{UB}^{\otimes n}\lVert\rho_U^{\otimes n}\otimes\rho_B^{\otimes n}\big)\right)}\\
    &\left\lVert\frac{1}{2^{R_J+R_K+R}}\sum_{(j,k,m)}\den{k}{k}\otimes\rho_{E\lvert U(j,k,m)}\right.\nonumber\\
    &\left.,\frac{1}{2^{R_K}}\sum_k\den{k}{k}\otimes\rho_0^{\otimes n}\right\rVert_1\nonumber\\
    &\le\frac{2(v_E)^{\frac{\alpha}{2}}}{\sqrt{\alpha}}2^{\frac{\alpha}{2}\left(-n(R_J+R)+\ubar{\D}_{1+\alpha}\big(\rho_{UE}^{\otimes n}\lVert\rho_U^{\otimes n}\otimes\rho_E^{\otimes n}\big)\right)}.
    \end{align}
\end{subequations}
Therefore, when $n\to\infty$ and $\alpha\to0$, from \eqref{eq:Final_PeCov_Asymp} there exists a sequence of codes such that \eqref{eq:Pe_Asymp} and \eqref{eq:CC_CSK_Metric_Asymp} hold if
\begin{subequations}\label{eq:Final_PeCov_Asymp_3}
\begin{align}
    R_J+R_K&>I(U;S),\label{eq:Encoding_Err}\\
    R_J+R_K+R&<I(U;B),\label{eq:Decoding_Err}\\
        R_J+R&>I(U;E).\label{eq:Covert_Cons}
\end{align}Applying Fourier-Motzkin elimination procedure \cite{ElGamalKim}, to eliminate $R_J$, in \eqref{eq:Final_PeCov_Asymp_3}, leads to the region in Corollary~\ref{lemma:Resolvability_1}.
\end{subequations}
\end{proof}
A consequence of Theorem~\ref{thm:Asymp} is the following two achievable rates for covert communication and covert secret key generation over a quantum state-dependent channel. In Corollary~\ref{Cor:Acivable_Covert_NC}, we project $\calR_{\textup{CC-CSK}}$ on the $R$ axis, i.e., setting $R_K=0$, and in Corollary~\ref{Cor:Acivable_SK_NC} we project $\calR_{\textup{CC-CSK}}$ on the $R_K$ axis, i.e., setting  $R=0$.
\begin{corollary}
\label{Cor:Acivable_Covert_NC}
The covert capacity of a quantum state-dependent channel $\calN_{AS\to BE}$, depicted in Fig.~\ref{fig:System_Model}, is lower bounded by
\begin{align}%
\calR_{\textup{CC}}\triangleq
\textup{sup}_{\rho_{UAS}} I(U;B)- I(U;S),
\label{eq:inRnone_R}
\end{align}where $\rho_{UAS}=\sum\limits_{u}p_{U}(u)\den{u}{u}_U\otimes\rho_{AS\lvert u}$ such that $\tra_{UA}[\rho_{UAS}]=\tra_{\bar{S}}[\rho_{\bar{S}S}]$, $I(U;B)\ge I(U;E)$, and $\rho_E=\rho_0$. 
    \end{corollary}
        \begin{corollary}
    \label{Cor:Acivable_SK_NC}
    The covert secret key capacity of a quantum state-dependent channel $\calN_{AS\to BE}$, depicted in Fig.~\ref{fig:System_Model}, is lower bounded by
    \begin{align}%
\calR_{\textup{CSK}}\triangleq
\textup{sup}_{\rho_{UAS}} I(U;B)- I(U;E),
\label{eq:inRnone_RK}
\end{align}where $\rho_{UAS}=\sum\limits_{u}p_{U}(u)\den{u}{u}_U\otimes\rho_{AS\lvert u}$ such that $\tra_{UA}[\rho_{UAS}]=\tra_{\bar{S}}[\rho_{\bar{S}S}]$, $I(U;B)\ge I(U;S)$, and $\rho_E=\rho_0$. 
\end{corollary}
We note that the problem of covert communication and covert secret key generation, studied in this paper, has not been studied for the classical case. However, a special case of this problem in which the transmitter only wants to communicate covertly with a positive rate via \ac{CSI} is studied in \cite{Covert_With_State,Keyless22}. We now show that when the channel is classical, Corollary~\ref{Cor:Acivable_Covert_NC} recovers the achievable rates in \cite{Covert_With_State} and \cite{Keyless22}.
\begin{spc}[Comparing with {\cite[Theorem~2]{Covert_With_State}}\hspace{-1mm} ]
The special case of Corollary~\ref{Cor:Acivable_Covert_NC} for classical channels recovers the achievable rate in \cite[Theorem~2]{Covert_With_State} if we set the rate of the external secret key shared between the transmitter and the receiver in \cite[Theorem~2]{Covert_With_State} equal to zero. 
\end{spc}

\begin{spc}[Comparing with {\cite[Theorem~4]{Keyless22}}\hspace{-1mm} ]    
We now show that the special case of Corollary~\ref{Cor:Acivable_Covert_NC} for classical channels is equivalent to the achievable rate in \cite[Theorem~4]{Keyless22}, which is based on generating a secret key from the \ac{CSI} and using the generated secret key to achieve a higher covert rate. First, note that if we set $V=\emptyset$ in \cite[Theorem~4]{Keyless22}, then we recover the special case of the achievable rate in Corollary~\ref{Cor:Acivable_Covert_NC} for classical channels. If we set $U=(\tilde{U},\tilde{V})$ in Corollary~\ref{Cor:Acivable_Covert_NC}, where $\tilde{U}$ and $\tilde{V}$ are independent but the \ac{CSI} is correlated with both of these auxiliary random variables and remove the tilde $\tilde{}$ to make the notation simpler, we have
    \begin{align}%
\calR_{\textup{CC}}=
\textup{sup}_{\rho_{UAS}} %
 I(U,V;B)- I(U,V;S),
\label{eq:inRnone_R2}
\end{align}where $\rho_{UVAS}=\sum\limits_u\sum\limits_vp_{U}(u)\den{u}{u}_U\otimes p_V(v)\den{v}{v}_V\otimes\rho_{AS\lvert u,v}$ such that $\tra_{UVA}[\rho_{UVAS}]=\tra_{\bar{S}}[\rho_{\bar{S}S}]$, $I(U,V;B)\ge I(U,V;E)$, and $\rho_E=\rho_0$. The achievable rate in \eqref{eq:inRnone_R2} for classical channels contains the achievable covert rate in \cite[Theorem~4]{Keyless22} because it has fewer constraints. Therefore, the  achievable rate in Corollary~\ref{Cor:Acivable_Covert_NC} for  classical channels is equivalent to~\cite[Theorem~4]{Keyless22}.
\end{spc}
\begin{spc}[Comparing with {\cite[Theorem~1]{Keyless22}}\hspace{-1mm} ]
    In Corollary~\ref{Cor:Acivable_Covert_NC}, when the channel is classical and the \ac{CSI} is available at both the transmitter and the receiver, if we set  $B\triangleq(\tilde{B},S)$ and $U=\triangleq(\tilde{U},S)$, where $\tilde{U}$ and $S$ are correlated, and $A=U$, then Corollary~\ref{Cor:Acivable_Covert_NC} recovers the covert capacity of state-dependent classical channels with the \ac{CSI} available at both the transmitter and receiver in \cite[Theorem~1]{Keyless22}.
\end{spc}

\subsection{Optimal Classical Results} 
We now show that the special case of Corollary~\ref{Cor:Acivable_SK_NC} for classical channels when the \ac{CSI} is available at both the transmitter and the receiver is optimal.

\begin{theorem}[Classical Covert Secret Key Capacity with Full \ac{CSI}]
    \label{thm:Capacitty_CSITR_C}
\begin{subequations}\label{eq:Capacitty_CSITR_C_AD}
The covert secret key capacity of the classical state-dependent \ac{DMC} $W_{BE\lvert AS}$, when the \ac{CSI} is available non-causally at both the transmitter and the receiver, is
\begin{align}
\textup{C}_{\textup{CSK}} =\mbox{\rm sup}_{P_{SUABE}\in\calD}\{I(U;B\lvert S)- I(U;E\lvert S)+ H(S\lvert E)\},
\label{eq:Capacitty_CSITR_C}
\end{align}
where
\begin{align}
  \calD \triangleq \left.\begin{cases}P_{SUABE}:\\
P_{SUABE}=Q_SP_{U\lvert S}P_{A\lvert US}W_{BE\lvert AS}\\
P_E=Q_0\\
\end{cases}\right\}.\label{eq:Capacitty_CSITR_C_D}
\end{align}
\end{subequations}
\end{theorem}
The achievability proof of Theorem~\ref{thm:Capacitty_CSITR_C} follows from Corollary~\ref{Cor:Acivable_SK_NC}, when the channel is classical and the \ac{CSI} is available at both the transmitter and the receiver, by setting $B\triangleq(\tilde{B},S)$, this is because the state is known at the receiver, and $U\triangleq(\tilde{U},S)$, where $\tilde{U}$ and $S$ are correlated, and removing the tilde $\tilde{}$ to make the notation simpler. In particular, 
\begin{subequations}\label{eq:Calculation_Classic}
    \begin{align}
        I(U;B)&=I(\tilde{U},S;\tilde{B},S)=I(\tilde{U};\tilde{B}\lvert S) +H(S),\\
        I(U;S)&=I(\tilde{U},S;S)=H(S),\\
        I(U;E)&=I(\tilde{U},S;E)=I(\tilde{U};E\lvert S)+H(S)-H(S\lvert E).
    \end{align}
\end{subequations}
We note that setting $U\triangleq(\tilde{U},S)$ can be interpreted as superpositioning the codebook on top of the \ac{CSI} $S$. The converse proof of Theorem~\ref{thm:Capacitty_CSITR_C} is similar to the converse proof of Theorem~\ref{thm:Capacitty_CSITR_CS} and is omitted for brevity.

The next result shows that the special case of Theorem~\ref{thm:Asymp} for the classical degraded channels when the \ac{CSI} is available non-causally at both the transmitter and the receiver is optimal.
\begin{theorem}[Classical Covert Communication and Covert Secret Key Generation Capacity Region with Full \ac{CSI} for Degraded Channels]
    \label{thm:Capacitty_CSITR_C_Degraded}
    The covert capacity of the classical state-dependent \ac{DMC} $W_{BE\lvert XS}$, when the warden's channel is degraded \ac{wrt} the legitimate receiver's channel and the \ac{CSI} is available non-causally at both the transmitter and the receiver,~is 
\begin{align}
  \calC_{\textup{CC-CSK}} = \left.\begin{cases}R,R_K\geq 0: \exists P_{SUABE}\in\calD:\\
R\le I(U;B\lvert S)\\
R_K\le I\left(U;B\lvert E,S\right)+ H(S\lvert E)\\
R+R_K\le I(U;B\lvert S)+ H(S)
\end{cases}\right\},\label{eq:Capacitty_CSITR_C_A_Degraded}
\end{align}
where $\calD$ is defined in \eqref{eq:Capacitty_CSITR_C_D} in Theorem~\ref{thm:Capacitty_CSITR_C}.
\end{theorem}
The achievability proof of Theorem~\ref{thm:Capacitty_CSITR_C_Degraded} follows from the achievability of Theorem~\ref{thm:Asymp}, when the channel is classical and degraded, and the \ac{CSI} is available at both the transmitter and the receiver, by setting $B\triangleq(\tilde{B},S)$ and $U\triangleq(\tilde{U},S)$, where $\tilde{U}$ and $S$ are correlated, and removing the tilde $\tilde{}$ to make the notation simpler. This follows from the calculations in \eqref{eq:Calculation_Classic}, which remain applicable here, along with the Markov chain $(\tilde{U},S)-\tilde{B}-\tilde{E}$, which implies $I(\tilde{U};\tilde{B}\lvert S)=I(\tilde{U};\tilde{B},\tilde{E}\lvert S)$. The converse proof of Theorem~\ref{thm:Capacitty_CSITR_C_Degraded} is provided in Appendix~\ref{App:Capacitty_CSITR_C_Degraded}. It is worth noting that the main challenge in proving a tight converse for general, non-degraded channels lies in dealing with the asymmetry in the security constraints for the message and the~key.
\section{Covert Secure  Communication, and Covert Secret Key Generation}
In this section, we first provide our one-shot and asymptotic results for the quantum state-dependent channels and then show that our results are optimal for the special case when the channel is classical and the state is known at both the transmitter and the receiver.
\subsection{One-Shot Results}Our one-shot result on the existence of an $\epsilon$-reliable and $\delta$-covert-secure $(2^R, 2^{R_K}, 1)$ code defined in Definition~\ref{defi:CSC_SKG} is as follows.
\begin{theorem}[Covert and Secure Communication and Covert Secret Key Generation]
\label{thm:One_Shot_CS}
Given any quantum state-dependent channel $\calN_{AS\to BE}$,  $\rho_{UAS}=\sum\limits_{u}p_{U}(u)\den{u}{u}_U\otimes\rho_{AS\lvert u}$, such that $\tra_{UA}[\rho_{UAS}]=\tra_{\bar{S}}[\rho_{\bar{S}S}]$, $\rho_E=\rho_0$, and $\rho_{UBE}\triangleq I_U\otimes\calN_{AS\to BE}(\rho_{UAS})$, there exists a $(2^R,2^{R_K},1)$ code such that
\begin{subequations}\label{eq:PeCov_Lemma_CS}
\begin{align}
    &\bbP\left\{(M,K)\ne(\hat{M},\hat{K})\right\}\nonumber\\
    &\le\frac{2v_S^\alpha}{\alpha}2^{\alpha\pr{-(R_J+R_K)+\ubar{\D}_{1+\alpha}\left(\rho_{US}\lVert\rho_U\otimes\rho_S\right)}}\nonumber\\
    &\quad+12v_B^\alpha2^{\alpha\pr{R_J+R_K+R-\ubar{\D}_{1-\alpha}\left(\rho_{UB}\lVert\rho_U\otimes\rho_B\right)}},\label{eq:Pe_Lemma_CS}\\
    &\left\lVert\hat{\rho}_{MKE}-\rho_M\otimes\sigma_K\otimes \rho_0\right\rVert_1\nonumber\\
    &\le\frac{2\pr{v_E}^{\frac{\alpha}{2}}}{\sqrt{\alpha}}2^{\frac{\alpha}{2}\pr{-R_J+\ubar{\D}_{1+\alpha}\left(\rho_{UE}\lVert\rho_U\otimes\rho_E\right)}},\label{eq:Covertnes_Lemma_CS}
    \end{align}where $R_J>0$ is an arbitrary integer, $\alpha\in\left(0,\frac{1}{2}\right)$, and $v_S$, $v_B$, and $v_E$ are the number of distinct eigenvalues of the states $\rho_S$, $\rho_B$ and $\rho_E$, respectively.
\end{subequations}
\end{theorem}
\begin{proof}
The codebook generation, encoding, and error analysis are similar to those in the proof of Theorem~\ref{thm:One_Shot} and are omitted for brevity. We now prove the covertness and security analysis. We note that in the problem studied in Section~\ref{sec:CC_CSK}, the message is not required to be secure. Consequently, the message $M$ and the local randomness $J$ in the codeword $U(J,K,M)$ play the role of the randomness required to keep the key secure against the warden. Note that the index $J$ also plays the role of the randomness required to correlate the codeword with the \ac{CSI} $\bar{S}$ as in the Gel'fand-Pinsker encoding. Whereas in the problem studied in this section, the message is both covert and secure, and therefore the local randomness $J$ in the codeword $U(J,K,M)$ plays the role of the part of the local randomness required to correlate the codeword with the \ac{CSI} $\bar{S}$ and it also plays the role of the randomness required to keep both the message and the key secure against the warden. The coding schemes for both Theorem~\ref{thm:One_Shot} and Theorem~\ref{thm:One_Shot_CS} are summarized in Fig.~\ref{fig:Coding_Scheme}. 

To prove the covertness and security of our code design, we bound $\bbE_C\Pu\left(\hat{\rho}_{MKE},\rho_M\otimes\sigma_K\otimes\rho_E\right)$, where $\hat{\rho}_{MKE}=\frac{1}{2^{R_J+R_K+R}}\sum_{(j,k,m)}\den{m}{m}\otimes\den{k}{k}\otimes\rho_{E\lvert U(j,k,m)}$ is the joint quantum state induced by our code design, and then choose $\rho_{UAS}$ such that $\rho_E=\rho_0$. We have \eqref{eq:Final_Covertness_CS} at the bottom of the next page,
\begin{figure*}[b!]
\hrulefill
\begin{align}
    &\bbE_C\Pu^2\left(\frac{1}{2^{R_J+R_K+R}}\sum_{(j,k,m)}\den{m}{m}\otimes\den{k}{k}\otimes\rho_{E\lvert U(j,k,m)},\frac{1}{2^{R+R_K}}\sum_{m\in\left[1:2^R\right]}\sum_{k\in\left[1:2^{R_K}\right]}\den{m}{m}\otimes\den{k}{k}\otimes\rho_E\right)\nonumber\\
    &\mathop\le\limits^{(a)}2-2\bbE_CF\left(\frac{1}{2^{R_J+R_K+R}}\sum_{(j,k,m)}\den{m}{m}\otimes\den{k}{k}\otimes\rho_{E\lvert U(j,k,m)},\frac{1}{2^{R+R_K}}\sum_{m\in\left[1:2^R\right]}\sum_{k\in\left[1:2^{R_K}\right]}\den{m}{m}\otimes\den{k}{k}\otimes\rho_E\right)\nonumber\\
    &\mathop=\limits^{(b)}2-2\bbE_C\left(\frac{1}{2^{R+R_K}}\sum_{m\in\left[1:2^R\right]}\sum_{k\in\left[1:2^{R_K}\right]}F\left(\frac{1}{2^{R_J}}\sum_j\rho_{E\lvert U(j,k,m)},\rho_E\right)\right)\nonumber\\
    &\mathop\le\limits^{(c)}\frac{2v_E^\alpha}{\alpha}2^{\alpha\left(-R_J+\ubar{\D}_{1+\alpha}\left(\rho_{UE}\lVert\rho_U\otimes\rho_E\right)\right)},\label{eq:Final_Covertness_CS}
\end{align}
\end{figure*}where 
 \begin{itemize}
    \item[$(a)$] follows from \eqref{eq:Purified_Bure} similar to \eqref{eq:Final_Covertness};
     \item[$(b)$] follows from the direct-sum property of the fidelity distance in \eqref{eq:Direct_Sum_Fidelity};
     \item[$(c)$] follows from $1-F(\rho,\sigma)\le\Pu^2(\rho,\sigma)$,  \eqref{eq:purified-sandReny}, \eqref{eq:Usefull_Inequality}, and Lemma~\ref{lemma:Resolvability_1}, similar to \eqref{eq:Final_Covertness}.
 \end{itemize}
Therefore, by considering \eqref{eq:Tr_Dist_Pu_Dist}, \eqref{eq:Final_Covertness_CS}, we obtain \eqref{eq:Covertnes_Lemma_CS}.
\end{proof}

\subsection{Asymptotic Results}  
The following result can be deduced from Theorem~\ref{thm:One_Shot_CS}, which establishes an achievable rate region for the asymptotic regime defined in Definition~\ref{defi:Asymp_CSC_CSK}. 
\begin{figure}[t!]
\centering
\includegraphics[width=9.0cm]{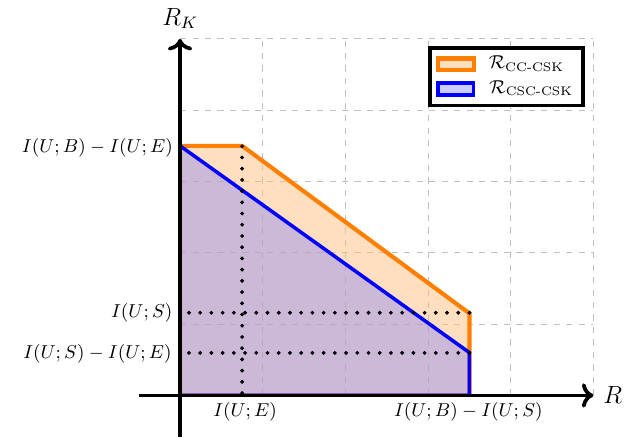}
\caption{The orange region illustrates the achievable rate region $\calR_{\textup{CC-CSK}}$ in Theorem~\ref{thm:Asymp} and the blue region illustrates the achievable rate region $\calR_{\textup{CSC-CSK}}$ in Theorem~\ref{thm:Asymp_CS}. }
\label{fig:Comparison}
\end{figure}
\begin{theorem}
\label{thm:Asymp_CS}
Define a rate region $\calR_{\textup{CSC-CSK}}$ as follows,
\begin{align}%
\calR_{\textup{CSC-CSK}}\triangleq
\bigcup_{\rho_{UAS}} %
\left\{ \begin{array}{rl}
  (R,R_K) \,:\;
	R &\leq I(U;B)- I(U;S)\\
        R+R_K &\leq I(U;B)- I(U;E)\\
	\end{array}
\right\},
\label{eq:inRnone_CS}
\end{align}where $\rho_{UAS}=\sum\limits_{u}p_{U}(u)\den{u}{u}_U\otimes\rho_{AS\lvert u}$ such that $\tra_{UA}[\rho_{UAS}]=\tra_{\bar{S}}[\rho_{\bar{S}S}]$ and $\rho_E=\rho_0$. An inner bound for the covert capacity region of a quantum state-dependent channel $\calN_{AS\to BE}$, depicted in Fig.~\ref{fig:System_Model}, is
    \begin{align}
        \calC_{\textup{CSC-CSK}}\supseteq\calR_{\textup{CSC-CSK}}.\nonumber
    \end{align}
\end{theorem}The proof of Theorem~\ref{thm:Asymp_CS} is similar to that of Theorem~\ref{thm:Asymp} and is omitted for brevity. The penalty term $I(U;S)$ represents the reduction in transmission rate required for the redundancy we need to add to correlate the codeword with the channel state, while the penalty term $I(U;E)$ reflects the loss in the transmission rate due to the imposed security constraint.

The achievable rate regions in Theorem~\ref{thm:Asymp} and Theorem~\ref{thm:Asymp_CS} are visualized in Fig.~\ref{fig:Comparison}. As seen in this figure, the achievable region $\calR_{\textup{CC-CSK}}$ in Theorem~\ref{thm:Asymp} is larger than the achievable region $\calR_{\textup{CSC-CSK}}$ in  Theorem~\ref{thm:Asymp_CS}. This is because the message in Theorem~\ref{thm:Asymp} is transmitted only covertly, while the message in Theorem~\ref{thm:Asymp_CS} is transmitted both securely and covertly. In other words, the constraint in \eqref{eq:CSC_CSK_Metric_Asymp} is more restrictive than the one in \eqref{eq:CC_CSK_Metric_Asymp}.

A consequence of Theorem~\ref{thm:Asymp_CS} is the following achievable rate for covert secure communication over a quantum state-dependent channel. In Corollary~\ref{Cor:Acivable_Covert_NC_CSC_CSK}, we project $\calR_{\textup{CSC-CSK}}$ on the $R$ axis, i.e., setting $R_K=0$. Note that the projection of $\calR_{\textup{CSC-CSK}}$ on the $R_K$ axis, i.e., setting  $R=0$, is the region presented in Corollary~\ref{Cor:Acivable_SK_NC}.
\begin{corollary}
\label{Cor:Acivable_Covert_NC_CSC_CSK}
The covert secure capacity of a quantum state-dependent channel $\calN_{AS\to BE}$, depicted in Fig.~\ref{fig:System_Model}, is lower bounded by
\begin{align}%
\calR_{\textup{CSC}}\triangleq
\textup{sup}_{\rho_{UAS}}\left[I(U;B) - \max\{I(U;S),I(U;E)\}\right],
\label{eq:inRnone_R_CSC_CSK}
\end{align}where $\rho_{UAS}=\sum\limits_{u}p_{U}(u)\den{u}{u}_U\otimes\rho_{AS\lvert u}$, such that $\tra_{UA}[\rho_{UAS}]=\tra_{\bar{S}}[\rho_{\bar{S}S}]$ and $\rho_E=\rho_0$. 
\end{corollary}

\subsection{Optimal Classical Results} 
We now show that the special case of Theorem~\ref{thm:Asymp_CS} for classical channels when the \ac{CSI} is available at both the transmitter and the receiver is optimal.
\begin{theorem}[Classical Covert Secure  Communication and Covert Secret Key Generation Capacity Region with Full \ac{CSI}]
    \label{thm:Capacitty_CSITR_CS}
    The covert secret capacity region of the classical state-dependent \ac{DMC} $W_{BE\lvert AS}$, when the \ac{CSI} is available non-causally at both the transmitter and the receiver, is
\begin{align}
  &\calC_{\textup{CSC-CSK}} =\nonumber\\
  &\left.\begin{cases}R,R_K\geq 0: \exists P_{SUABE}\in\calD:\\
R\le I(U;B\lvert S)\\
R+R_K\le I(U;B\lvert S)- I(U;E\lvert S)+H(S\lvert E)
\end{cases}\right\},\label{eq:Capacitty_CSITR_CS_A}
\end{align}
where $\calD$ is defined in \eqref{eq:Capacitty_CSITR_C_D} in Theorem~\ref{thm:Capacitty_CSITR_C}.
\end{theorem}
Similar to the achievability proofs of Theorem~\ref{thm:Capacitty_CSITR_C} and Theorem~\ref{thm:Capacitty_CSITR_C_Degraded}, the achievability proof of Theorem~\ref{thm:Capacitty_CSITR_CS} follows from Theorem~\ref{thm:Asymp_CS} by setting $B\triangleq(\tilde{B},S)$ and $U\triangleq(\tilde{U},S)$, where $\tilde{U}$ and $S$ are correlated, and removing the tilde\, $\tilde{}$\, to make the notation simpler. The converse proof of Theorem~\ref{thm:Capacitty_CSITR_CS} is provided in Appendix~\ref{App:Capacitty_CSITR_CS}.

\begin{remark}
   A consequence of Theorem~\ref{thm:Capacitty_CSITR_CS} is the characterization of the optimal covert secure communication rate over a state-dependent channel with state available at both the transmitter and the receiver, which corresponds to the projection of $\calC_{\textup{CSC-CSK}}$ on the $R$ axis, i.e., setting $R_K=0$; and the optimal covert secret key generation rate over a state-dependent channel with state available at both the transmitter and the receiver, which corresponds to the projection of $\calC_{\textup{CSC-CSK}}$ on the $R_K$ axis. Note that the optimal covert secret key generation rate is equivalent to that in Theorem~\ref{thm:Capacitty_CSITR_C}.
\end{remark}
\begin{remark}[Stealth Communications and Key Generation over Quantum State-Dependent  Channel]
The notion of stealth communication in classical channels was introduced in \cite{effectivesecrec}. For the quantum state-dependent channel in Fig.~\ref{fig:System_Model}, if we modify the security metric in Definitions~\ref{defi:CC_SKG},  Definition~\ref{defi:CSC_SKG}, Definition~\ref{defi:Asymp_CC}, and Definition~\ref{defi:Asymp_CSC_CSK}, respectively, as
    \begin{align*}
    \left\lVert\hat{\rho}_{KE}-\sigma_K\otimes \rho_E\right\rVert_1&\le\delta,\nonumber\\
    \left\lVert\hat{\rho}_{MKE}-\rho_M\otimes\sigma_K\otimes \rho_E\right\rVert_1&\le\delta,\\
    \lim_{n\to\infty}\left\lVert\hat{\rho}_{KE^n}-\sigma_K\otimes\rho_E^{\otimes n}\right\rVert_1&=0,\\
        \lim_{n\to\infty}\left\lVert\hat{\rho}_{MKE^n}-\rho_M\otimes\sigma_K\otimes\rho_E^{\otimes n}\right\rVert_1&=0,
    \end{align*}where $\rho_E=\tra_{B}[\bbI\otimes\calN_{AS\to BE}(\rho_{USA})]$. All our results in this paper still hold if we remove the covertness constraint $\rho_E=\rho_0$ (or $P_E=Q_0$ for the classical results) from all of our results. This follows from our achievability proofs, where we first induce the state $\rho_E$ at the channel output of the warden, and then we choose the state $\rho_{UAS}$ such that $\rho_E=\rho_0$. Hence, our results also provide fundamental limits of stealth communication and stealth secret key generation over quantum state-dependent channels.  
\end{remark}
\section{Conclusions and Open Problems}
This paper studies (i) covert communication and covert secret key generation, and  (ii) covert secure communication and covert secret key generation over quantum state-dependent channels. We provide one-shot achievability regions for both of these problems and show that these achievability regions are tight when the channel is classical and the state is available at both the transmitter and the receiver.

The problem of covert communication and secret key generation over quantum state-dependent channels, when the \ac{CSI} is available causally at the transmitter, is left as future work. A possible achievability scheme for both covert communication and covert secret key generation could leverage a block Markov encoding scheme to utilize the \ac{CSI}, however, studying the covertness constraint in such a block Markov encoding scheme remains an open problem. Note that if we do not consider covert secret key generation, one can derive the following achievable covert communication rate without using block Markov encoding when \ac{CSI} is available at the transmitter causally, as an entangled state with the channel.
\begin{theorem}
\label{thm:Asymp_Causal}
    The covert capacity of a state-dependent quantum channel $\calN_{AS\to BE}^{\otimes n}$, depicted in Fig.~\ref{fig:System_Model}, when the state is known causally at the transmitter, is lower bounded by
    \begin{align}
        \sup_{p_U,W_{US\to A}}I(U;B),\nonumber
    \end{align}such that $I(U;B)>I(U;E)$, and $\rho_E=\rho_0$. 
\end{theorem}Theorem~\ref{thm:Asymp_Causal} is proved in Appendix~\ref{Achi:Scheme_Causal}. 
\begin{remark}[Comparing with  Classical Results]
      Theorem~\ref{thm:Asymp_Causal} is a quantum generalization of the classical achievable rates in \cite[Theorem~8]{Keyless22} and  in \cite[Theorem~2]{Covert_With_State}, with $R_K=0$.
\end{remark}

\begin{appendices}

\section{Proof of Lemma~\ref{lemma:Resolvability_1}}
\label{proof:lemma:Resolvability_1}
The proof follows the approach of \cite[Lemma~6]{AshnuHayashi2020}. From the concavity of the $\log$ function and Jensen's inequality, we have
\begin{align}
    \bbE_C\left[\ubar{\D}_{1+\alpha}\big(\tau_{E\lvert C}\lVert\rho_E\big)\right]&\le\frac{1}{\alpha}\log_2\left(\bbE_C\left[2^{\alpha\ubar{\D}_{1+\alpha}\big(\tau_{E\lvert C}\lVert\rho_E\big)}\right]\right).\label{eq:lemma1_first}
\end{align}
Then, we bound the argument of the $\log$ function in \eqref{eq:lemma1_first} as in \eqref{eq:Lemma1_Final} at the bottom of the next page, 
\begin{figure*}[bh!]
\hrulefill
\begin{align}
    &\bbE_C\left[2^{\alpha\ubar{\D}_{1+\alpha}\big(\tau_{E\lvert C}\lVert\rho_E\big)}\right]\nonumber\\
    &=\bbE_C\tra\left[\left(\rho_E^{-\frac{\alpha}{2(1+\alpha)}}\tau_{E\lvert C}\rho_E^{-\frac{\alpha}{2(1+\alpha)}}\right)^{1+\alpha}\right]\nonumber\\
    &=\bbE_C\tra\left[\left(\rho_E^{-\frac{\alpha}{2(1+\alpha)}}\frac{1}{2^R}\sum_{i}\rho_{E\lvert U(i)}\rho_E^{-\frac{\alpha}{2(1+\alpha)}}\right)^{1+\alpha}\right]\nonumber\\
    &=\bbE_C\tra\left[\left(\rho_E^{-\frac{\alpha}{2(1+\alpha)}}\frac{1}{2^R}\sum_{i}\rho_{E\lvert U(i)}\rho_E^{-\frac{\alpha}{2(1+\alpha)}}\right)\left(\rho_E^{-\frac{\alpha}{2(1+\alpha)}}\frac{1}{2^R}\sum_{i'}\rho_{E\lvert U(i')}\rho_E^{-\frac{\alpha}{2(1+\alpha)}}\right)^{\alpha}\right]\nonumber\\
    &=\frac{1}{2^R}\bbE_C\tra\left[\sum_{i}\left(\rho_E^{-\frac{\alpha}{2(1+\alpha)}}\rho_{E\lvert U(i)}\rho_E^{-\frac{\alpha}{2(1+\alpha)}}\right)\left(\rho_E^{-\frac{\alpha}{2(1+\alpha)}}\frac{1}{2^R}\sum_{i'}\rho_{E\lvert U(i')}\rho_E^{-\frac{\alpha}{2(1+\alpha)}}\right)^{\alpha}\right]\nonumber\\
    &\mathop=\limits^{(a)}\frac{1}{2^R}\sum_{i}\bbE_C\tra\left[\left(\rho_E^{-\frac{\alpha}{2(1+\alpha)}}\rho_{E\lvert U(i)}\rho_E^{-\frac{\alpha}{2(1+\alpha)}}\right)\left(\rho_E^{-\frac{\alpha}{2(1+\alpha)}}\frac{1}{2^R}\sum_{i'}\rho_{E\lvert U(i')}\rho_E^{-\frac{\alpha}{2(1+\alpha)}}\right)^{\alpha}\right]\nonumber\\
    &=\frac{1}{2^R}\sum_{i}\bbE_C\tra\left[\left(\rho_E^{-\frac{\alpha}{2(1+\alpha)}}\rho_{E\lvert U(i)}\rho_E^{-\frac{\alpha}{2(1+\alpha)}}\right)\left(\rho_E^{-\frac{\alpha}{2(1+\alpha)}}\frac{1}{2^R}\left(\rho_{E\lvert U(i)}+\sum_{i'\ne i}\rho_{E\lvert U(i')}\right)\rho_E^{-\frac{\alpha}{2(1+\alpha)}}\right)^{\alpha}\right]\nonumber\\
    &\mathop\le\limits^{(b)}\frac{1}{2^R}\sum_{i}\bbE_{U(i)}\tra\left[\left(\rho_E^{-\frac{\alpha}{2(1+\alpha)}}\rho_{E\lvert U(i)}\rho_E^{-\frac{\alpha}{2(1+\alpha)}}\right)\left(\rho_E^{-\frac{\alpha}{2(1+\alpha)}}\frac{1}{2^R}\left(\rho_{E\lvert U(i)}+\sum_{i'\ne i}\bbE_{U(i')}\rho_{E\lvert U(i')}\right)\rho_E^{-\frac{\alpha}{2(1+\alpha)}}\right)^{\alpha}\right]\nonumber\\
    &=\frac{1}{2^R}\sum_{i}\bbE_{U(i)}\tra\left[\left(\rho_E^{-\frac{\alpha}{2(1+\alpha)}}\rho_{E\lvert U(i)}\rho_E^{-\frac{\alpha}{2(1+\alpha)}}\right)\left(\rho_E^{-\frac{\alpha}{2(1+\alpha)}}\frac{1}{2^R}\left(\rho_{E\lvert U(i)}+\left(2^R-1\right)\rho_E\right)\rho_E^{-\frac{\alpha}{2(1+\alpha)}}\right)^{\alpha}\right]\nonumber\\
    &\mathop\le\limits^{(c)}\frac{1}{2^R}\sum_{i}\bbE_{U(i)}\tra\left[\left(\rho_E^{-\frac{\alpha}{2(1+\alpha)}}\rho_{E\lvert U(i)}\rho_E^{-\frac{\alpha}{2(1+\alpha)}}\right)\left(\rho_E^{-\frac{\alpha}{2(1+\alpha)}}\frac{1}{2^R}\left(v_E\calE_{\rho_E}\left(\rho_{E\lvert U(i)}\right)+\left(2^R-1\right)\rho_E\right)\rho_E^{-\frac{\alpha}{2(1+\alpha)}}\right)^{\alpha}\right]\nonumber\\
    &=\frac{1}{2^R}\sum_{i}\bbE_{U(i)}\tra\!\left[\left(\rho_E^{-\frac{\alpha}{2(1+\alpha)}}\rho_{E\lvert U(i)}\rho_E^{-\frac{\alpha}{2(1+\alpha)}}\right)\!
 \left(\rho_E^{-\frac{\alpha}{2(1+\alpha)}}\frac{1}{2^R}v_E\calE_{\rho_E}\left(\rho_{E\lvert U(i)}\right)\rho_E^{-\frac{\alpha}{2(1+\alpha)}}\!
 +\frac{\left(2^R-1\right)}{2^R}\rho_E^{-\frac{\alpha}{2(1+\alpha)}}\rho_E\rho_E^{-\frac{\alpha}{2(1+\alpha)}}\right)^{\alpha}\right]\nonumber\\
    &\mathop\le\limits^{(d)}\frac{1}{2^R}\sum_{i}\bbE_{U(i)}\tra\left[\left(\rho_E^{-\frac{\alpha}{2(1+\alpha)}}\rho_{E\lvert U(i)}\rho_E^{-\frac{\alpha}{2(1+\alpha)}}\right)\left(\rho_E^{-\frac{\alpha^2}{2(1+\alpha)}}\frac{1}{2^{\alpha R}}\left(v_E^\alpha\left(\calE_{\rho_E}\left(\rho_{E\lvert U(i)}\right)\right)^\alpha+\left(2^R-1\right)^\alpha\rho_E^\alpha\right)\rho_E^{-\frac{\alpha^2}{2(1+\alpha)}}\right)\right]\nonumber\\
    &\mathop\le\limits^{(e)}1+\bbE_U\tra\left[\frac{v_E^\alpha}{2^{\alpha R}}\rho_{E\lvert U}\left(\calE_{\rho_E}\left(\rho_{E\lvert U}\right)\right)^\alpha\rho_E^{-\alpha}\right]\nonumber\\
    &\mathop=\limits^{(f)}1+\bbE_U\tra\left[\frac{v_E^\alpha}{2^{\alpha R}}\rho_{E\lvert U}\calE_{\rho_E}\left(\left(\calE_{\rho_E}\left(\rho_{E\lvert U}\right)\right)^\alpha\rho_E^{-\alpha}\right)\right]\nonumber\\
    &\mathop=\limits^{(g)}1+\tra\left[\frac{v_E^\alpha}{2^{\alpha R}}\bbE_U\left(\calE_{\rho_E}\left(\rho_{E\lvert U}\right)\right)^{1+\alpha}\rho_E^{-\alpha}\right]\nonumber\\
    &\mathop=\limits^{(h)}1+\frac{v_E^\alpha}{2^{\alpha R}}\tra\left[\left(\calE_{\rho_U\otimes\rho_E}\left(\rho_{UE}\right)\right)^{1+\alpha}\left(\rho_U\otimes\rho_E\right)^{-\alpha}\right]\nonumber\\
    &\mathop\le\limits^{(i)}1+\frac{v_E^\alpha}{2^{\alpha R}}2^{\alpha\ubar{\D}_{1+\alpha}\big(\rho_{UE}\lVert\rho_U\otimes\rho_E\big)},\label{eq:Lemma1_Final}
\end{align}
\end{figure*}where
\begin{itemize}
    \item[$(a)$] follows from the linearity of the trace operation and the expectation;
    \item[$(b)$] follows from the linearity of the expectation, the law of total expectations, Jensen's inequality, and since the symbols $U(i)$ and $U(i')$, for $i\ne i'$, are independent;
    \item[$(c)$] follows since for two quantum states $\sigma\in\calD(\calH)$ and $\rho\in\calD(\calH)$ we have \cite{Hayashi02},
    \begin{align}
        \rho\preceq v\calE_\sigma(\rho),\label{eq:Pinching_Inequlaity}
    \end{align}where $v$ is the distinct number of eigenvalues of $\sigma$;
    \item[$(d)$] follows since the terms inside the trace operation commute and since for $a,b\in\bbR^+$ we have $(a+b)^x\le a^x+b^x$, for $x<1$;
    \item[$(e)$] follows from the linearity of the trace operation, the symmetry of the codebook construction \ac{wrt} the index $i$, and since $\calE_{\rho_E}\left(\rho_{E\lvert U}\right)$ and $\rho_E$ commute;
    \item[$(f)$] follows since $\calE_{\rho_E}\left(\rho_{E\lvert U}\right)$ has the same orthonormal basis as $\rho_E$ and therefore $\left(\calE_{\rho_E}\left(\rho_{E\lvert U}\right)\right)^\alpha\rho_E^{-\alpha}=\calE_{\rho_E}\left(\left(\calE_{\rho_E}\left(\rho_{E\lvert U}\right)\right)^\alpha\rho_E^{-\alpha}\right)$;
    \item[$(g)$] follow since for three arbitrary states $\sigma\in\calD(\calH)$, with spectral decomposition $\sigma=\sum_i\lambda_i\den{x_i}{x_i}$, and $\rho_1\in\calD(\calH)$, and $\rho_2\in\calD(\calH)$ we have
    \begin{align}
    \tra\left[\calE_\sigma(\rho_1)\rho_2\right]&=\sum_i\bra{x_i}\rho_1\ket{x_i}\tra\left[\den{x_i}{x_i}\rho_2\right]\nonumber\\
    &=\sum_i\bra{x_i}\rho_1\ket{x_i}\tra\left[\bra{x_i}\rho_2\ket{x_i}\right]\nonumber\\
    &=\sum_i\tra\left[\bra{x_i}\rho_1\ket{x_i}\right]\bra{x_i}\rho_2\ket{x_i}\nonumber\\
    &=\sum_i\tra\left[\rho_1\den{x_i}{x_i}\right]\bra{x_i}\rho_2\ket{x_i}\nonumber\\
    &=\sum_i\tra\left[\rho_1\bra{x_i}\rho_2\ket{x_i}\den{x_i}{x_i}\right]\nonumber\\
    &=\tra[\rho_1\calE_\sigma(\rho_2)];\label{eq:Pinching_Switching}
    \end{align}
    \item[$(h)$] follows since the involved states are classical-quantum states;
    \item[$(i)$] follows from the definition of $\ubar{\D}_{1+\alpha}\big(\cdot\lVert\cdot\big)$ and the data processing inequality. 
\end{itemize}Now substituting \eqref{eq:Lemma1_Final} in \eqref{eq:lemma1_first} and considering the fact that $\log_2(1+x)\le\frac{x}{\ln2}$ completes the proof of Lemma~\ref{lemma:Resolvability_1}.

\section{Proof of Lemma~\ref{lemma:Decodability}}
\label{proof:lemma:Decodability}
We first prove \eqref{eq:Decoding_Lemma_1}, as provided in \eqref{eq:lemma_reliability_1} at the bottom of the next page,
\begin{figure*}[bh!]
\hrulefill
\begin{align}
    &\tra\left[\left(\bbI-\Pi_{UB}\right)\rho_{UB}\right]\nonumber\\
    &\mathop=\limits^{(a)}\tra\left[\left(\calE_{\rho_B}\left(\bbI-\Pi_{UB}\right)\right)\rho_{UB}\right]\nonumber\\
    &\mathop=\limits^{(b)}\tra\left[\left(\bbI-\Pi_{UB}\right)\calE_{\rho_B}(\rho_{UB})\right]\nonumber\\
    &=\tra\left[\left(\bbI-\Pi_{UB}\right)\left(\calE_{\rho_B}(\rho_{UB})\right)^{1-\alpha}\left(\calE_{\rho_B}(\rho_{UB})\right)^\alpha\right]\nonumber\\
    &\mathop\le\limits^{(c)}2^{\alpha(R_J+R_K+R)}\tra\left[\left(\bbI-\Pi_{UB}\right)\left(\calE_{\rho_B}(\rho_{UB})\right)^{1-\alpha}\left(\rho_U\otimes\rho_B\right)^\alpha\right]\nonumber\\
    &\mathop\le\limits^{(d)}2^{\alpha(R_J+R_K+R)}\tra\left[\left(\calE_{\rho_B}(\rho_{UB})\right)^{1-\alpha}\left(\rho_U\otimes\rho_B\right)^\alpha\right]\nonumber\\
    &=2^{\alpha(R_J+R_K+R)}\tra\left[\left(\left(\rho_U\otimes\rho_B\right)^{\frac{\alpha}{2(1-\alpha)}}\calE_{\rho_B}(\rho_{UB})\left(\rho_U\otimes\rho_B\right)^{\frac{\alpha}{2(1-\alpha)}}\right)^{1-\alpha}\right]\nonumber\\
    &=2^{\alpha(R_J+R_K+R)}\tra\left[\left(\left(\rho_U\otimes\rho_B\right)^{\frac{\alpha}{2(1-\alpha)}}\calE_{\rho_B}(\rho_{UB})\left(\rho_U\otimes\rho_B\right)^{\frac{\alpha}{2(1-\alpha)}}\right)\left(\left(\rho_U\otimes\rho_B\right)^{\frac{\alpha}{2(1-\alpha)}}\calE_{\rho_B}(\rho_{UB})\left(\rho_U\otimes\rho_B\right)^{\frac{\alpha}{2(1-\alpha)}}\right)^{-\alpha}
    \right]\nonumber\\
    &\mathop=\limits^{(e)}2^{\alpha(R_J+R_K+R)}\tra\left[\left(\left(\rho_U\otimes\rho_B\right)^{\frac{\alpha}{2(1-\alpha)}}\rho_{UB}\left(\rho_U\otimes\rho_B\right)^{\frac{\alpha}{2(1-\alpha)}}\right)\left(\left(\rho_U\otimes\rho_B\right)^{\frac{\alpha}{2(1-\alpha)}}\calE_{\rho_B}(\rho_{UB})\left(\rho_U\otimes\rho_B\right)^{\frac{\alpha}{2(1-\alpha)}}\right)^{-\alpha}
    \right]\nonumber\\
    &\mathop\le\limits^{(f)}v_B^\alpha2^{\alpha(R_J+R_K+R)}\tra\left[\left(\left(\rho_U\otimes\rho_B\right)^{\frac{\alpha}{2(1-\alpha)}}\rho_{UB}\left(\rho_U\otimes\rho_B\right)^{\frac{\alpha}{2(1-\alpha)}}\right)\left(\left(\rho_U\otimes\rho_B\right)^{\frac{\alpha}{2(1-\alpha)}}\rho_{UB}\left(\rho_U\otimes\rho_B\right)^{\frac{\alpha}{2(1-\alpha)}}\right)^{-\alpha}
    \right]\nonumber\\
    &=v_B^\alpha2^{\alpha(R_J+R_K+R)}\tra\left[\left(\left(\rho_U\otimes\rho_B\right)^\frac{\alpha}{2(1-\alpha)}\rho_{UB}\left(\rho_U\otimes\rho_B\right)^\frac{\alpha}{2(1-\alpha)}\right)^{1-\alpha}\right]\nonumber\\
    &=v_B^\alpha2^{\alpha(R_J+R_K+R)}2^{-\alpha\ubar{\D}_{1-\alpha}\big(\rho_{UB}\lVert\rho_U\otimes\rho_B\big)},\label{eq:lemma_reliability_1}
    \end{align}
    \hrulefill
    \setcounter{equation}{48}
    \begin{align}
    &\tra\left[\Pi_{UB}\left(\rho_U\otimes\rho_B\right)\right]\nonumber\\
    &=\tra\left[\Pi_{UB}\left(\rho_U\otimes\rho_B\right)^{1-\alpha}\left(\rho_U\otimes\rho_B\right)^\alpha\right]\nonumber\\
    &\mathop\le\limits^{(a)}2^{-(1-\alpha)(R_J+R_K+R)}\tra\left[\Pi_{UB}\left(\calE_{\rho_B}(\rho_{UB})\right)^{1-\alpha}\left(\rho_U\otimes\rho_B\right)^\alpha\right]\nonumber\\
    &\mathop\le\limits^{(b)}2^{-(1-\alpha)(R_J+R_K+R)}\tra\left[\left(\calE_{\rho_B}(\rho_{UB})\right)^{1-\alpha}\left(\rho_U\otimes\rho_B\right)^\alpha\right]\nonumber\\
    &=2^{-(1-\alpha)(R_J+R_K+R)}\tra\left[\left(\left(\rho_U\otimes\rho_B\right)^{\frac{\alpha}{2(1-\alpha)}}\calE_{\rho_B}(\rho_{UB})\left(\rho_U\otimes\rho_B\right)^{\frac{\alpha}{2(1-\alpha)}}\right)^{1-\alpha}\right]\nonumber\\
    &=2^{-(1-\alpha)(R_J+R_K+R)}\tra\!\left[\left(\left(\rho_U\otimes\rho_B\right)^{\frac{\alpha}{2(1-\alpha)}}\calE_{\rho_B}(\rho_{UB})\left(\rho_U\otimes\rho_B\right)^{\frac{\alpha}{2(1-\alpha)}}\right)\!\left(\!\left(\rho_U\otimes\rho_B\right)^{\frac{\alpha}{2(1-\alpha)}}\calE_{\rho_B}(\rho_{UB})\left(\rho_U\otimes\rho_B\right)^{\frac{\alpha}{2(1-\alpha)}}\right)^{-\alpha}
    \right]\nonumber\\
    &\mathop=\limits^{(c)}2^{-(1-\alpha)(R_J+R_K+R)}\tra\left[\left(\left(\rho_U\otimes\rho_B\right)^{\frac{\alpha}{2(1-\alpha)}}\rho_{UB}\left(\rho_U\otimes\rho_B\right)^{\frac{\alpha}{2(1-\alpha)}}\right)\left(\left(\rho_U\otimes\rho_B\right)^{\frac{\alpha}{2(1-\alpha)}}\calE_{\rho_B}(\rho_{UB})\left(\rho_U\otimes\rho_B\right)^{\frac{\alpha}{2(1-\alpha)}}\right)^{-\alpha}
    \right]\nonumber\\
    &\mathop\le\limits^{(d)}v_B^\alpha2^{-(1-\alpha)(R_J+R_K+R)}\tra\left[\left(\left(\rho_U\otimes\rho_B\right)^{\frac{\alpha}{2(1-\alpha)}}\rho_{UB}\left(\rho_U\otimes\rho_B\right)^{\frac{\alpha}{2(1-\alpha)}}\right)\left(\left(\rho_U\otimes\rho_B\right)^{\frac{\alpha}{2(1-\alpha)}}\rho_{UB}\left(\rho_U\otimes\rho_B\right)^{\frac{\alpha}{2(1-\alpha)}}\right)^{-\alpha}
    \right]\nonumber\\
    &=v_B^\alpha2^{-(1-\alpha)(R_J+R_K+R)}\tra\left[\left(\left(\rho_U\otimes\rho_B\right)^\frac{\alpha}{2(1-\alpha)}\rho_{UB}\left(\rho_U\otimes\rho_B\right)^\frac{\alpha}{2(1-\alpha)}\right)^{1-\alpha}\right]\nonumber\\
    &=v_B^\alpha2^{-(1-\alpha)(R_J+R_K+R)}2^{-\alpha\ubar{\D}_{1-\alpha}\big(\rho_{UB}\lVert\rho_U\otimes\rho_B\big)},\label{eq:lemma_reliability_2}
\end{align}
\setcounter{equation}{47}
\end{figure*}
    where
\begin{itemize}
    \item[$(a)$] follows since $\bbI-\Pi_{UB}$ commutes with $\rho_U\otimes\rho_B$ and therefore $\calE_{\rho_B}\left(\left(\bbI-\Pi_{UB}\right)\right)=\bbI-\Pi_{UB}$;
    \item[$(b)$] follows from \eqref{eq:Pinching_Switching};
    \item[$(c)$] follows from the definition of $\Pi_{UB}$ in \eqref{eq:Main_Projection} and since $f(x)=x^{s}$, for $s\in(0\,,1]$, is a matrix monotone function \cite[Section~1.5]{QIT_Hayashi};
    \item[$(d)$] follows since for $A,B,C\in\calP(\calH)$ and $C\le B$ we have $\tra(CA)\le\tra(BA)$, since $\tra((B-C)A)\ge0$ \cite[Lemma~B.5.2]{RennerDissertation};
    \item[$(e)$] follows since for two arbitrary states $\sigma\in\calD(\calH)$, with spectral decomposition $\sigma=\sum_i\lambda_i\den{x_i}{x_i}$, and $\rho\in\calD(\calH)$ we have $\tra\left[\calE_\sigma(\rho)\sigma\right]=\tra\left[\rho\sigma\right]$ since
        \begin{align}
        \tra[\rho\sigma]&=\sum_i\lambda_i\tra\left[\rho\den{x_i}{x_i}\right]\nonumber\\
        &=\sum_i\lambda_i\tra\left[\bra{x_i}\rho\ket{x_i}\right]\nonumber\\
        &=\tra\left[\calE_\sigma(\rho)\sigma\right];\label{eq:Pinching_Property_2}
    \end{align}
    \item[$(f)$] follows from \eqref{eq:Pinching_Inequlaity} and since for $A,B,C\in\calP(\calH)$ and $C\le B$ we have $\tra(CA)\le\tra(BA)$ \cite[Lemma~B.5.2]{RennerDissertation}, since for a hermitian matrix $S$ and $\rho\in\calP(\calH)$, $S\rho S$ is non-negative \cite[Lemma~B.5.1]{RennerDissertation}, and since $f(x)=x^{-s}$, for $s\in(0\,,1]$, is a matrix anti-monotone function \cite[Section~1.5]{QIT_Hayashi}.
\end{itemize}
    We now prove \eqref{eq:Decoding_Lemma_2}, as provided in \eqref{eq:lemma_reliability_2} at the bottom of the previous page, where
    \setcounter{equation}{49}
\begin{itemize}
    \item[$(a)$] follows from the definition of $\Pi_{UB}$ in \eqref{eq:Main_Projection} and since $f(x)=x^{s}$, for $s\in(0\,,1]$, is a matrix monotone function \cite[Section~1.5]{QIT_Hayashi};
    \item[$(b)$] follows since $\Pi_{UB}\le\bbI$ and for $A,B,C\in\calP(\calH)$ and $C\le B$ we have $\tra(CA)\le\tra(BA)$, since $\tra((B-C)A)\ge0$ \cite[Lemma~B.5.2]{RennerDissertation};
    \item[$(c)$] follows from \eqref{eq:Pinching_Property_2};
    \item[$(d)$] follows from \eqref{eq:Pinching_Inequlaity} and since for $A,B,C\in\calP(\calH)$ and $C\le B$ we have $\tra(CA)\le\tra(BA)$ \cite[Lemma~B.5.2]{RennerDissertation}, since for a hermitian matrix $S$ and $\rho\in\calP(\calH)$, $S\rho S$ is non-negative \cite[Lemma~B.5.1]{RennerDissertation}, and since $f(x)=x^{-s}$, for $s\in(0\,,1]$, is a matrix anti-monotone function \cite[Section~1.5]{QIT_Hayashi}.
\end{itemize}

\section{Proof of Theorem~\ref{thm:Capacitty_CSITR_C_Degraded}}
\label{App:Capacitty_CSITR_C_Degraded}
Consider any sequence of codes with length $n$ for state-dependent channels when the \ac{CSI} is available non-causally at both the transmitter and the receiver such that $P_e^{(n)}\le\epsilon_n$, where $\epsilon_n\xrightarrow[]{n\to\infty}0$, $ I(K;E^n)\le\eta_n$, where $\eta_n\xrightarrow[]{n\to\infty}0$, and $\bbD\left(P_{E^n}\lVert Q_0^{\otimes n}\right)\le\tilde{\epsilon}$. 
We first define $\calA^{(\epsilon)}$ for $\epsilon>0$ which expands the rate region defined in \eqref{eq:Capacitty_CSITR_C_A_Degraded} as
\begin{subequations}\label{eq:Converse_AD_epsilon_C_Degraded}
\begin{align}
  \calA^{(\epsilon)}= \left.\begin{cases}R,R_K\geq 0: \exists P_{SUABE}\in\calD^{(\epsilon)}:\\
  R\le I(U;B\lvert S)+\epsilon\\
  R_K\le I(U;B\lvert E,S)+H(S\lvert E)+\epsilon\\
  R+R_K\le I(U;B\lvert S)+H(S)+\epsilon\\
\end{cases}\right\},\label{eq:A_epsilon_C_CSITR_Degraded} \displaybreak[0]
\end{align}
where
\begin{align}
  \calD^{(\epsilon)}\triangleq \left.\begin{cases}P_{SUABE}:\\
P_{SUABE}=Q_SP_{U\lvert S}P_{X\lvert US}W_{BE\lvert AS}\\
\bbD\big(P_E\lVert Q_0\big)\le\epsilon\\
\end{cases}\right\}.\label{eq:D_epsilon_C_CSITR_Degraded}
\end{align}
\end{subequations}We now show that any achievable rate pair $(R,R_K)$ belongs to $\calA^{(\epsilon)}$. 
For any $\lambda,\mu>0$, and $\epsilon_n>0$ we have,
\begin{align}
    &nR_K\le H(K\lvert E^n)+\eta_n\nonumber\\
    &\mathop\le\limits^{(a)}H(K\lvert E^n)-H(K\lvert B^n,S^n)+n\epsilon_n+\eta_n\nonumber\\
    &\le H(K\lvert E^n)-H(K\lvert B^n,S^n,E^n)+n\epsilon_n+\eta_n\nonumber\\
    &= I(K;B^n,S^n\lvert E^n)+n\epsilon_n+\eta_n\nonumber\\
    &\le I(K;B^n\lvert E^n,S^n)+H(S^n\lvert E^n)+n\epsilon_n+\eta_n\nonumber\\ 
    &=\sum_{t=1}^n\left[ I(K;B_t\lvert E^n,S^n,B^{t-1})+H(S_t\lvert E^n,S^{t-1})\right]\nonumber\\
    &\qquad+n\epsilon_n+\eta_n\nonumber\\
    &\le\sum_{t=1}^n\left[\bbI\left(K,M,S^{t-1},S_{t+1}^n,B^{t-1},E^{t-1},E_{t+1}^n;B_t\lvert E_t,S_t\right)\right.\nonumber\\
    &\left.\qquad+H(S_t\lvert E_t)\right]+n\epsilon_n+\eta_n\nonumber\\
    &\mathop\le\limits^{(b)}\sum_{t=1}^n\left[\bbI\left(U_t;B_t\lvert E_t,S_t\right)+H(S_t\lvert E_t)\right]+n\epsilon_n+\eta_n\nonumber\\
    &=n\left[\bbI\left(U_T;B_T\lvert E_T,S_T,T\right)+H(S_T\lvert E_T,T)\right]+n\epsilon_n+\eta_n\nonumber\\
    &\le n\left[\bbI\left(U_T,T;B_T\lvert E_T,S_T\right)+H(S_T\lvert E_T)\right]+n\epsilon_n+\eta_n\nonumber \\
    &\mathop=\limits^{(c)} n\left[I(U;B\lvert E,S)+H(S\lvert E)+\epsilon_n+\frac{\eta_n}{n}\right]\nonumber\\
    &\mathop\le\limits^{(d)} n\left[I(U;B\lvert E,S)+H(S\lvert E)+2\epsilon\right],\label{eq:RK_Converse_Degr2}
\end{align}where
\begin{itemize}
    \item[$(a)$] follows from Fano's inequality;
    \item[$(b)$] follows by defining  \begin{align}
        U_i&\triangleq\left(K,M,S^{t-1},S_{t+1}^n,B^{t-1},E^{t-1},E_{t+1}^n\right);\label{eq:Ut_Degraded}
    \end{align}
    \item[$(c)$] follows by defining
    \begin{align}
        U\triangleq(U_T,T),\,S\triangleq S_T,\,B\triangleq B_T,\,E\triangleq E_T;\label{eq:TS_RVs_Degraded}
    \end{align}
     \item[$(d)$] follows by defining $\epsilon\triangleq\max\{\epsilon_n,\lambda,\mu\}$, in which we choose $n$ large enough such that $\lambda\ge\frac{\eta_n}{n}$ and $\mu\ge\frac{2\tilde{\epsilon}}{n}$.
\end{itemize}

We also have
\begin{align}
    &nR=H(M)\nonumber\\
    &\mathop\le\limits^{(a)} I(M;B^n,S^n)+n\epsilon_n\nonumber\\
    &\mathop=\limits^{(b)} I(M;B^n\lvert S^n)+n\epsilon_n\nonumber\\
    &=\sum_{t=1}^n I(M;B_t\lvert S^n,B^{t-1})+n\epsilon_n\label{eqc52}\\
    &\le\sum_{t=1}^n I(K,M,S^{t-1},S_{t+1}^n,B^{t-1},E^{t-1},E_{t+1}^n;B_t\lvert S_t)+n\epsilon_n\nonumber\\
    &\mathop=\limits^{(c)} I(U_t;B_t\lvert S_t)+n\epsilon_n\nonumber\\
    &\mathop\le\limits^{(d)}n I(U;B\lvert S)+n\epsilon,\label{eq:Upper_R_Degraded}
\end{align}where
\begin{itemize}
    \item[$(a)$] follows from Fano's inequality;
    \item[$(b)$] follows since the message and \ac{CSI} are independent;
    \item[$(c)$] follows from \eqref{eq:Ut_Degraded};
    \item[$(d)$] follows from \eqref{eq:TS_RVs_Degraded} and by defining $\epsilon$ as in \eqref{eq:RK_Converse_Degr2}. 
\end{itemize}

We also have
\begin{align}
    &n(R+R_K)=H(M,K)\nonumber\\
    &\mathop\le\limits^{(a)} I(M,K;B^n,S^n)+n\epsilon_n\nonumber\\
    &\le I(M,K;B^n\lvert S^n)+H(S^n)+n\epsilon_n\nonumber\\
    &\mathop=\limits^{(b)}\sum_{t=1}^n \left[I(M,K;B_t\lvert S^n,B^{t-1})+H(S_t)\right]+n\epsilon_n\nonumber \\
    &\le\sum_{t=1}^n \left[I(K,M,S^{t-1},S_{t+1}^n,B^{t-1},E^{t-1},E_{t+1}^n;B_t\lvert S_t)\right.\nonumber\\
    &\left.\qquad+H(S_t)\right]+n\epsilon_n\nonumber\\
    &\mathop=\limits^{(c)}\sum_{t=1}^n \left[I(U_t;B_t\lvert S_t)+H(S_t)\right]+n\epsilon_n\nonumber\\
    &\mathop\le\limits^{(d)}n\left[ I(U;B\lvert S)+H(S)\right]+n\epsilon,\label{eq:Converse_Sum_Degraded}
\end{align}where
\begin{itemize}
    \item[$(a)$] follows from Fano's inequality;
    \item[$(b)$] follows since the \ac{CSI} is \ac{iid};
    \item[$(c)$] follows from \eqref{eq:Ut_Degraded}; 
    \item[$(d)$] follows from \eqref{eq:TS_RVs_Degraded} and the definition of $\epsilon$ in \eqref{eq:RK_Converse_Degr2}. 
\end{itemize}Then, we have 
\begin{align}
\bbD(P_E\lVert Q_0)&=\bbD(P_{E_T}\lVert Q_0)\nonumber\\
&=\bbD\Bigg(\frac{1}{n}\sum\limits_{t=1}^nP_{E_t}\bigg\lVert Q_0\Bigg)\nonumber\\
&\leq\frac{1}{n}\sum\limits_{t=1}^n\bbD(P_{E_t}\lVert Q_0)\nonumber\\
&\leq\frac{1}{n}\bbD(P_{E^n}\lVert Q_0^{\otimes n})\nonumber\\
&\leq\frac{\tilde{\epsilon}}{n}\nonumber\\
&\leq\epsilon,\label{eq:Covertness_C_Degraded}
\end{align}where the first inequality follows from Jensen's Inequality. 
Combining \eqref{eq:RK_Converse_Degr2}, \eqref{eq:Upper_R_Degraded}, \eqref{eq:Converse_Sum_Degraded}, and \eqref{eq:Covertness_C_Degraded} shows that 
$\calC_{\textup{CC-CSK}} \subseteq\bigcap_{\epsilon>0}\mathcal{A}_{\epsilon}.$ 
The proof of continuity at zero of $\calA_\epsilon$ is similar to that of \cite[Appendix~F]{Keyless22} and is omitted.

\section{Converse Proof of Theorem~\ref{thm:Capacitty_CSITR_CS}}
\label{App:Capacitty_CSITR_CS}
Consider any sequence of codes with length $n$ for state-dependent channels when the \ac{CSI} is available non-causally at both the transmitter and the receiver such that $P_e^{(n)}\le\epsilon_n$, where $\epsilon_n\xrightarrow[]{n\to\infty}0$, $ I(M,K;E^n)\le\eta_n$, where $\eta_n\xrightarrow[]{n\to\infty}0$, and $\bbD\left(P_{E^n}\lVert Q_0^{\otimes n}\right)\le\tilde{\epsilon}$. 
 We first define $\calA^{(\epsilon)}$ for $\epsilon>0$ which expands the rate region defined in \eqref{eq:Capacitty_CSITR_CS_A} as
\begin{subequations}\label{eq:Converse_AD_epsilon_C}
\begin{align}
  &\calA^{(\epsilon)}=\nonumber\\
  &\left.\begin{cases}R,R_K\geq 0: \exists P_{SUABE}\in\calD^{(\epsilon)}:\\
  R\le I(U;B\lvert S)+\epsilon\\
  R+R_K\le I(U;B\lvert S)- I(U;E\lvert S)+H(S\lvert E)+\epsilon\\
\end{cases}\right\},\label{eq:A_epsilon_CS_CSITR}
\end{align}
where
\begin{align}
  \calD^{(\epsilon)}\triangleq \left.\begin{cases}P_{SUABE}:\\
P_{SUABE}=Q_SP_{U\lvert S}P_{X\lvert US}W_{BE\lvert AS}\\
\bbD\big(P_E\lVert Q_0\big)\le\epsilon\\
\end{cases}\right\}.\label{eq:D_epsilon_CS_CSITR}
\end{align}
\end{subequations}We now show that any achievable rate pair $(R,R_K)$ belongs to $\calA^{(\epsilon)}$. 
For any $\lambda,\mu>0$, $\epsilon_n>0$, we~have
\begin{align}
    nR&=H(M)\nonumber\\
    &\mathop\le\limits^{(a)}
    \sum_{t=1}^n I(M;B_t\lvert S^n,B^{t-1})+n\epsilon_n\nonumber\\
    &\leq\sum_{t=1}^n I(K,M,S^{t-1},S_{t+1}^n,B^{t-1},E_{t+1}^n;B_t\lvert S_t)+n\epsilon_n\nonumber\\
    &\mathop\le\limits^{(b)}n I(U;B\lvert S)+n\epsilon,\label{eq:Upper_R_CS}
\end{align}where
\begin{itemize}
    \item[$(a)$] follows similar to \eqref{eqc52};
    \item[$(b)$] follows with 
    \begin{subequations}
    \begin{align}
        U_i&\triangleq\left(K,M,S^{t-1},S_{t+1}^n,B^{t-1},E_{t+1}^n\right),\label{eq:Ut_CS}\\
        U&\triangleq(U_T,T),\,S\triangleq S_T,\,B\triangleq B_T,\,E\triangleq E_T,\label{eq:TS_RVs_CS}
        \end{align}
        and $\epsilon\triangleq\max\{\epsilon_n,\lambda,\mu\}$, where we choose $n$ large enough such that $\lambda\ge\frac{\eta_n}{n}$ and $\mu\ge\frac{2\tilde{\epsilon}}{n}$.
        \end{subequations}
\end{itemize}
\begin{figure*}[bh!]
\hrulefill
\setcounter{equation}{63}
\begin{align}
    &\bbE_{C_n}\bbD\left(\frac{1}{2^{nR}}\sum_{m}\rho_{E^n\lvert U^n(m)}\Big\lVert\rho_E^{\otimes n}\right)\nonumber \displaybreak[0]\\
    &=\bbE_{C_n}\tra\left(\frac{1}{2^{nR}}\sum_m\rho_{E^n\lvert U^n(m)}\left[\log\left(\frac{1}{2^{nR}}\sum_{m'}\rho_{E^n\lvert U^n(m')}\right)-\log\left(\rho_E^{\otimes n}\right)\right]\right)\nonumber\\
    &=\frac{1}{2^{nR}}\tra\left(\sum_m\bbE_{C_n}\rho_{E^n\lvert U^n(m)}\left[\log\left(\frac{1}{2^{nR}}\rho_{E^n\lvert U^n(m)}+\frac{1}{2^{nR}}\sum_{m'\ne m}\rho_{E^n\lvert U^n(m')}\right)-\log\left(\rho_E^{\otimes n}\right)\right]\right)\nonumber\\
    &\mathop\le\limits^{(a)}\frac{1}{2^{nR}}\tra\left(\sum_m\bbE_{U^n(m)}\rho_{E^n\lvert U^n(m)}\left[\log\left(\frac{1}{2^{nR}}\rho_{E^n\lvert U^n(m)}+\frac{1}{2^{nR}}\sum_{m'\ne m}\bbE_{U^n(m')}\rho_{E^n\lvert U^n(m')}\right)-\log\left(\rho_E^{\otimes n}\right)\right]\right)\nonumber\\
    &=\frac{1}{2^{nR}}\tra\left(\sum_m\bbE_{U^n(m)}\rho_{E^n\lvert U^n(m)}\left[\log\left(\frac{1}{2^{nR}}\rho_{E^n\lvert U^n(m)}+\frac{\left(2^{nR}-1\right)}{2^{nR}}\rho_E^{\otimes n}\right)-\log\left(\rho_E^{\otimes n}\right)\right]\right)\nonumber\\
    &\mathop\le\limits^{(b)}\frac{1}{2^{nR}}\tra\left(\sum_m\bbE_{U^n(m)}\rho_{E^n\lvert U^n(m)}\log\left(I+\frac{1}{2^{nR}}\rho_{E^n\lvert U^n(m)}\left(\rho_E^{\otimes n}\right)^{-1}\right)\right)\nonumber\\
    &\mathop\le\limits^{(c)}\frac{1}{2^{nR}}\tra\left(\sum_m\bbE_{U^n(m)}\rho_{E^n\lvert U^n(m)}\log\left(I+\frac{v_E}{2^{nR}}\calE_{\rho_E^{\otimes n}}\pr{\rho_{E^n\lvert U^n(m)}}\left(\rho_E^{\otimes n}\right)^{-1}\right)\right)\nonumber\\
    &=\frac{1}{\alpha2^{nR}}\tra\left(\sum_m\bbE_{U^n(m)}\rho_{E^n\lvert U^n(m)}\log\left(I+\frac{v_E}{2^{nR}}\calE_{\rho_E^{\otimes n}}\pr{\rho_{E^n\lvert U^n(m)}}\left(\rho_E^{\otimes n}\right)^{-1}\right)^\alpha\right)\nonumber\\
    &\mathop\le\limits^{(d)}\frac{1}{\alpha2^{nR}}\tra\left(\sum_m\bbE_{U^n(m)}\rho_{E^n\lvert U^n(m)}\frac{v_E^\alpha}{2^{\alpha nR}}\calE_{\rho_E^{\otimes n}}\pr{\rho_{E^n\lvert U^n(m)}}^\alpha\left(\rho_E^{\otimes n}\right)^{-\alpha}\right)\nonumber\\
    &\mathop=\limits^{(e)}\frac{v_E^\alpha}{\alpha2^{n\alpha R}}\tra\left(\bbE_{U^n}\rho_{E^n\lvert U^n}\calE_{\rho_E^{\otimes n}}\pr{\rho_{E^n\lvert U^n}}^{\alpha}\left(\rho_E^{\otimes n}\right)^{-\alpha}\right)\nonumber\\
    &=\frac{v_E^\alpha}{\alpha2^{n\alpha R}}\tra\left(\bbE_{U^n}\rho_{E^n\lvert U^n}\calE_{\rho_E^{\otimes n}}\pr{\calE_{\rho_E^{\otimes n}}\pr{\rho_{E^n\lvert U^n}}^{\alpha}\left(\rho_E^{\otimes n}\right)^{-\alpha}}\right)\nonumber\\
    &\mathop=\limits^{(f)}\frac{v_E^\alpha}{\alpha2^{n\alpha R}}\tra\left(\bbE_{U^n}\calE_{\rho_E^{\otimes n}}\pr{\rho_{E^n\lvert U^n}}^{1+\alpha}\left(\rho_E^{\otimes n}\right)^{-\alpha}\right)\nonumber\\
    &\mathop=\limits^{(g)}\frac{v_E^\alpha}{\alpha2^{n\alpha R}}\tra\left(\calE_{\rho_{U^n}\otimes\rho_E^{\otimes n}}\pr{\rho_{U^nE^n}}^{1+\alpha}\left(\rho_{U^n}\otimes\rho_E^{\otimes n}\right)^{-\alpha}\right)\nonumber\\
    &\mathop\le\limits^{(h)}\frac{v_E^\alpha}{\alpha}2^{\alpha(-nR+\ubar{\D}_{1+\alpha}(\rho_{U^nE^n};\rho_{U^n}\otimes \rho_E^{\otimes n}))}\nonumber\\
    &\mathop\le\limits^{(i)}\frac{v_E^\alpha}{\alpha}2^{\alpha n(-R+\ubar{\D}_{1+\alpha}(\rho_{UE};\rho_{U}\otimes \rho_{E}))},\label{eq:Resolv_Analysis}
\end{align}
\setcounter{equation}{60}
\end{figure*}
We also have,
\begin{align}
    &n(R+R_K)=H(M,K)\nonumber\\
    &\mathop\le\limits^{(a)}H(M,K)-H(M,K\lvert B^n,S^n)+n\epsilon_n\nonumber\\
    &\mathop\le\limits^{(b)} I(M,K;B^n,S^n)- I(M,K;E^n)+n\epsilon_n+\eta_n\nonumber\\
    &= I(M,K,S^n;B^n,S^n)- I(S^n;B^n,S^n\lvert M,K)\nonumber\\
    &\qquad - I(M,K,S^n;E^n)+I(S^n;E^n\lvert M,K)+n\epsilon_n+\eta_n\nonumber\\
    &= I(M,K,S^n;B^n,S^n)-H(S^n\lvert M,K)- I(M,K,S^n;E^n)\nonumber\\
    &\qquad+ I(S^n;E^n\lvert M,K)+n\epsilon_n+\eta_n\nonumber\\
    &\le I(M,K,S^n;B^n,S^n)- I(M,K,S^n;E^n)+n\epsilon_n+\eta_n\nonumber\\
    &=\sum_{t=1}^n\left[ I(M,K,S^n;B_t,S_t\lvert B^{t-1},S^{t-1})\right.\nonumber\\
    &\left.\qquad- I(M,K,S^n;E_t\lvert E_{t+1}^n)\right]+n\epsilon_n+\eta_n\nonumber\\
    &\mathop=\limits^{(c)}\sum_{t=1}^n\left[ I(M,K,S^n,E_{t+1}^n;B_t,S_t\lvert B^{t-1},S^{t-1})\right.\nonumber\\
    &\left.\qquad- I(M,K,S^n,B^{t-1},S^{t-1};E_t\lvert E_{t+1}^n)\right]+n\epsilon_n+\eta_n\nonumber\\
    &\le\sum_{t=1}^n\left[ I(M,K,S^n,E_{t+1}^n,B^{t-1},S^{t-1};B_t,S_t)\right.\nonumber\\
    &\left.\qquad- I(M,K,S^n,B^{t-1},S^{t-1};E_t\lvert E_{t+1}^n)\right]+n\epsilon_n+\eta_n\nonumber\\
    &\mathop\le\limits^{(d)}\sum_{t=1}^n\left[ I(M,K,S^n,E_{t+1}^n,B^{t-1},S^{t-1};B_t,S_t)\right.\nonumber\\
    &\left.\qquad- I(M,K,S^n,B^{t-1},S^{t-1},E_{t+1}^n;E_t)\right]+n\epsilon_n+\eta_n+\tilde{\epsilon}\nonumber\\
    &\mathop=\limits^{(e)}\sum_{t=1}^n\left[ I(U_t,S_t;B_t,S_t)- I(U_t,S_t;E_t)\right]\nonumber\\
    &\qquad+n\epsilon_n+\eta_n+\tilde{\epsilon}\nonumber\\
    &=n\left[ I(U_T,S_T;B_T,S_T\lvert T)- I(U_T,S_T;E_T\lvert T)\right]\nonumber\\
    &\qquad+n\epsilon_n+\eta_n+\tilde{\epsilon}\nonumber\\
    &\le n\left[ I(U_T,S_T,T;B_T,S_T)- I(U_T,S_T;E_T\lvert T)\right]\nonumber\\
    &\qquad+n\epsilon_n+\eta_n+\tilde{\epsilon}\nonumber\\
    &\mathop\le\limits^{(f)}n\left[ I(U_T,S_T,T;B_T,S_T)- I(U_T,S_T,T;E_T)\right]\nonumber\\
    &\qquad+n\epsilon_n+\eta_n+2\tilde{\epsilon}\nonumber\\
    &\mathop=\limits^{(g)}n\left[ I(U,S;B,S)- I(U,S;E)\right]+n\epsilon_n+\eta_n+2\tilde{\epsilon}\nonumber\\
    &=n\left[ I(U;B\lvert S)- I(U;E\lvert S)+H(S\lvert E)\right]+n\epsilon_n+\eta_n+2\tilde{\epsilon},\nonumber\\
    &= n\left[I(U;B\lvert S)- I(U;E\lvert S)+H(S\lvert E)+\epsilon_n+\frac{\eta_n}{n}+2\frac{\tilde{\epsilon}}{n}\right]\nonumber\\
    &\mathop\le\limits^{(h)} n\left[I(U;B\lvert S)- I(U;E\lvert S)+ H(S\lvert E)+3\epsilon\right],\label{eq:SumRate2_CS}
\end{align}where
\begin{itemize}
    \item[$(a)$] follows from Fano's inequality;
    \item[$(b)$] follows from the security constraint;
    \item[$(c)$] follows from  \cite[Lemma~7]{BCC:IT78};
    \item[$(d)$] and $(f)$ follow from \cite[Lemma~3]{Keyless22};
    \item[$(e)$] follows from \eqref{eq:Ut_CS};
    \item[$(g)$] follows from \eqref{eq:TS_RVs_CS};
    \item[$(h)$] follows by defining $\epsilon$ and in \eqref{eq:Upper_R_CS}.
\end{itemize}
Similar to \eqref{eq:Covertness_C_Degraded}, one can show that $\bbD\big(P_E\lVert Q_0\big)\le\epsilon$, 
which combined with \eqref{eq:Upper_R_CS} and \eqref{eq:SumRate2_CS}, shows that 
$\calC_{\textup{CSC-CSK}} \subseteq \bigcap_{\epsilon>0}\mathcal{A}_{\epsilon}.$  The proof of continuity at zero of $\calA_\epsilon$ is similar to that of \cite[Appendix~F]{Keyless22} and is omitted.

\section{Proof of Theorem~\ref{thm:Asymp_Causal}}
\label{Achi:Scheme_Causal}
Fix the distribution $p_U$, and the isometry $W_{US\to A}$, which takes states $U$ and  $S$ as inputs and outputs the state $A$.
\subsection{Codebook Generation}Let $C_n\triangleq\big(U^n(m)\big)_{m\in\calM}$, where $\calM\triangleq\sbra{1}{2^{nR}}$, be a random codebook generated \ac{iid} according~to $p_U$ 
and $\calC_n\triangleq\big(u^n(m)\big)_{m\in\calM}$ be a realization of the codebook $C_n$.

\subsection{Encoding and Decoding}
Given the message $m$ and the state $\bar{S}^i$, the encoder computes the codeword $U^n(m)$ first, and then applies the isometry $W_{U_i(m)\bar{S}_i\to A_i}$ on $U_i(m)$ and $\bar{S}_i$ and transmits $A_i$ over~$\calN_{AS\to BE}$, for $i\in\sbra{1}{n}$. 
Then, for any $\epsilon > 0$, by the Packing~Lemma~{\cite[Lemma~16.3.1]{Wilde_Book}}, there exists a set of \ac{POVM} operators $\{\Lambda_m\}_{m \in \calM_n}$ and a constant $c > 0$ such that, averaged over the random codebook ensemble, the probability of error is upper bounded by 
\begin{align}
 &1-\bbE_{C_n}\left[\frac{1}{\abs{\calM_n}}\sum_{m\in\calM_n}\tra\left(\Lambda_m\rho_{B^n}^{U^n(m)}\right)\right]\nonumber\\
 &\le2(\epsilon+2\sqrt{\epsilon})+4\times2^{n(R-I(U;B)+2c\epsilon)},\label{eq:PError}
\end{align}where the inequality holds by \cite[Lemma~16.3.1]{Wilde_Book} and the \ac{RHS} of \eqref{eq:PError} vanishes as $n$ grows~when
\begin{align}
    R<I(U;B). \label{eq:Decodability}
\end{align}

\subsection{Covertness Analysis}
The following lemma is essential in the covertness analysis.
\begin{lemma}
    \label{lemma:Resolvability_C}
    Let $\rho_{UE}\triangleq\sum_{u\in\calU}p_U(u)\den{u}{u}\otimes\rho_{E\lvert u}$ be a classical-quantum state. Also, let $C_n\triangleq\{U^n(m)\}_{m\in\calM_n}$, where $\calM_n\triangleq\left[2^{nR}\right]$, be a set of random variables in which $U^n(m)$ is generated \ac{iid} according to $p_U$. Then, for $R>I(U;E)$,
    \begin{align}
        \bbE_{C_n}\bbD\left(\frac{1}{2^{nR}}\sum_{m}\rho_{E^n\lvert U^n(m)}\Big\lVert\rho_E^{\otimes n}\right)\xrightarrow[n\to\infty]{}0.\nonumber
    \end{align}
\end{lemma}
\begin{proof}
    See Appendix~\ref{proof:lemma:Resolvability_C}.
\end{proof}
To prove the covertness of our code design, we bound $\bbE_{C_n}\bbD\left(\hat{\rho}_{E^n}\lVert\rho_E^{\otimes n}\right)$, where $\hat{\rho}_{E^n}$ is the distribution induced by our code design and is $\hat{\rho}_{E^n}=\frac{1}{2^{nR}}\sum_m\rho_{E^n\lvert U^n(m)}$, and then choose $p_U$ and $W_{US\to A}$ such that $\rho_E=\rho_0$. Now from Lemma~\ref{lemma:Resolvability_C} and \cite[Theorem~11.9.1]{Wilde_Book}, we have
\begin{align} \bbE_{C_n}\left\lVert\hat{\rho}_{E^n}-\rho_E^{\otimes n}\right\rVert_1\xrightarrow[n\to\infty]{}0,\nonumber
\end{align}if $R>I(U;Z)$, which with \eqref{eq:Decodability} completes the proof of Theorem~\ref{thm:Asymp_Causal}.

\section{Proof of Lemma~\ref{lemma:Resolvability_C}}
\label{proof:lemma:Resolvability_C}
We have, \eqref{eq:Resolv_Analysis} provided at the bottom of the previous page, 
where
\begin{itemize}
    \item[$(a)$] follows from the linearity of the expectation, the law of total expectations, Jensen's inequality, and since the symbols $U(i)$ and $U(i')$, for $i\ne i'$, are independent;
    \item[$(b)$] follows since $\log$ is a matrix monotone function;
    \item[$(c)$] follows from \eqref{eq:Pinching_Inequlaity} and $v_E$ is upper bounded as in \eqref{eq:vS_vB};
    \item[$(d)$] follows since $\calE_{\rho_E^{\otimes n}}\pr{\rho_{E^n\lvert U^n(m)}}$ commutes with $\rho_E^{\otimes n}$ and for $a,b\in\bbR^+$, $0<\alpha\le1$, we have $\frac{1}{\alpha}\log(1+x)^\alpha\le\frac{1}{\alpha}\log(1^\alpha+x^\alpha)\le\frac{x^\alpha}{\alpha}$;
    \item[$(e)$] follows from the symmetry of the codebook construction \ac{wrt} the message;
    \item[$(f)$] follows from \eqref{eq:Pinching_Switching};
    \item[$(g)$] follows since the states are classical-quantum;
    \item[$(h)$] follows from the definition of the sandwiched R\'{e}nyi relative entropy and since for two  states $\rho,\sigma\in\calD(\calH)$ and $\calE(\cdot):\calL(A)\to\calL(B)$, we have $\ubar{\D}_{1+\alpha}\pr{\calE(\rho);\calE(\sigma)}\le\ubar{\D}_{1+\alpha}\pr{\rho;\sigma}$;
    \item[$(i)$] follows since from the codebook construction $U^n$ is an \ac{iid} sequence.
\end{itemize}Therefore, when $n\to\infty$ and $\alpha\to0$ the \ac{RHS} of \eqref{eq:Resolv_Analysis} vanishes if $ R>I(U;E)$.

\end{appendices}
\bibliographystyle{IEEEtran}
\bibliography{IEEEabrv,bibfile}

\end{document}